\documentclass[12pt, draftclsnofoot, onecolumn]{IEEEtran}

\usepackage[cmex10]{amsmath}
\usepackage{amssymb}
\usepackage{wasysym}

\DeclareSymbolFont{rsfs}{U}{rsfs}{m}{n}
\DeclareSymbolFontAlphabet{\mathscrsfs}{rsfs}

\usepackage{mathtools}
\usepackage{bm}
\usepackage{subcaption}
\usepackage{amsthm}

\usepackage{tikz}
\usetikzlibrary{arrows,positioning,shapes.geometric,calc}
\usepackage{pgfplots}
\pgfplotsset{compat=1.15}

\usepackage{comment}

\DeclarePairedDelimiter{\ceil}{\lceil}{\rceil}
\DeclarePairedDelimiter{\floor}{\lfloor}{\rfloor}

\definecolor{dgreen}{RGB}{0,100,0}
\newtheoremstyle{custom}
{} % Space above
{} % Space below
{} % Body font
{} % Indent amount
{\bfseries} % Theorem head font
{:} % Punctuation after theorem head
{.25em} % Space after theorem head
{} % Theorem head spec (can be left empty, meaning `normal')
\theoremstyle{custom}

\newtheorem{theorem}{Theorem}
\newtheorem{lemma}{Lemma}

\newtheorem{definition}{Definition}

\newtheorem*{theorem*}{Theorem}
\newtheorem*{lemma*}{Lemma}
\newtheorem*{proposition*}{Proposition}
\newtheorem*{definition*}{Definition}
\newtheorem*{example*}{Example}
\newtheorem*{remark*}{Remark}
\newtheorem*{corollary*}{Corollary}

\DeclareMathOperator*{\arctanh}{arctanh}

\makeatletter
\let\l@ENGLISH\l@english
\makeatother

\title{\textbf{Parallelism versus Latency in Simplified\\[-0.72ex]
\hspace*{-2.34ex}%
\mbox{Successive-Cancellation Decoding of Polar Codes}\\[0.36ex]}}

\author{Seyyed~Ali~Hashemi, Marco~Mondelli, Arman~Fazeli, Alexander~Vardy, John~Cioffi, and Andrea~Goldsmith 
\thanks{S.~A.~Hashemi and J.~Cioffi are with the Department of Electrical Engineering, Stanford University, Stanford, CA 94305, USA (email: ahashemi@stanford.edu, cioffi@stanford.edu). M.~Mondelli is with the Institute of Science and Technology (IST) Austria, Klosterneuburg, Austria (email: marco.mondelli@ist.ac.at). A.~Fazeli and A.~Vardy are with the Department of Electrical and Computer Engineering, UC San Diego, La Jolla, CA 92093, USA (email: afazelic@ucsd.edu, avardy@ucsd.edu). A.~Goldsmith is with the Department of Electrical Engineering, Princeton University, Princeton, NJ 08544, USA (email: goldsmith@princeton.edu).}
}

\begin{document}
\bstctlcite{IEEEexample:BSTcontrol}
\maketitle
%------------------------------------------------------------------------------------------------------------------

\vspace{-3.60ex}
\begin{abstract}
\noindent 
\looseness=-1
This paper characterizes the latency of the simplified successive-cancellation (SSC) decoding scheme for polar codes under hardware resource constraints. In particular, when the number of processing~elements~$P$ that can perform SSC decoding operations in parallel is limited, as is the case in practice, the latency of SSC decoding is $O\left(N^{1-1/\mu}+\frac{N}{P}\log_2\log_2\frac{N}{P}\right)$, where $N$ is the block length of the code and $\mu$ is the scaling exponent of the channel. Three direct consequences of this bound are presented. First, in a \mbox{fully-parallel} implementation where $P=\frac{N}{2}$, the latency of SSC decoding is $O\left(N^{1-1/\mu}\right)$, which is sublinear in the block length. This recovers a result from our earlier work. Second, in a \mbox{fully-serial} implementation where $P=1$, the latency of SSC decoding scales as $O\left(N\log_2\log_2 N\right)$. The multiplicative constant is also calculated: we show that the latency of SSC decoding when $P=1$ is given by $\left(2+o(1)\right) N\log_2\log_2 N$. Third, in a \mbox{semi-parallel} implementation, the smallest $P$ that gives the same latency as that of the fully-parallel implementation is $P=N^{1/\mu}$. The tightness of our bound on SSC decoding latency and the applicability of the foregoing results is validated through extensive simulations.
\end{abstract}

%\begin{IEEEkeywords}
%Polar codes
%\end{IEEEkeywords}

\vspace{3.60ex}
%----------------------------------------------------
\section{Introduction} \label{sec:intro}
%----------------------------------------------------

% polar codes under SC decoding
\noindent
Polar codes \cite{Ari09} have been adopted as the coding scheme for control and physical broadcast channels of the enhanced mobile broadband (eMBB) mode and the ultra-reliable low latency communications (URLLC) mode in the fifth generation (5G) wireless communications standard \cite{3gpp_polar, WON2020}. For a polar code of block length $N$, the encoding and successive-cancellation (SC) decoding complexity for any binary memoryless symmetric (BMS) channel is $O\left(N\log_2 N\right)$. Polar codes can be constructed with complexity that is sublinear in $N$ \cite{mondelli2018construction}, and the error probability under SC decoding scales with the block length roughly as $2^{-\sqrt{N}}$ \cite{ArT09}. The gap to capacity scales with the block length roughly as 
\begin{equation}
I(W)-R ~\sim~ N^{-\mu} \ \text{,}
\end{equation}
where $W$ is the BMS transmission channel, $I(W)$ is its capacity, $R$ is the rate of the code, and $\mu$ is called the \emph{scaling exponent} (see \cite{HAU14, MHU15unif-ieeeit, XG13, GB14, MHU14list-ieeeit, fazeli2018binary, guruswami2019ar}). In general, the scaling exponent $\mu$ depends on the transmission channel $W$. It is known \cite{HAU14, MHU15unif-ieeeit} that $3.579 \le \mu \le 4.714$ for any BMS channel $W$. Furthermore, $\mu \approx 3.63$ when $W$ is a binary erasure channel (BEC), as shown in \cite{MHU15unif-ieeeit}, $\mu \approx 4$  when $W$ is a binary additive white Gaussian noise channel (BAWGNC), as shown in \cite{KMTU10}, and it is conjectured that $\mu \approx 4.2$ when $W$ is a binary symmetric channel (BSC). It is possible to approach the optimal scaling exponent $\mu=2$ by using large polarization kernels \cite{fazeli2018binary, guruswami2019ar, wang2019polar}. The moderate deviations regime, in which both the error probability and the gap to capacity jointly vanish as the block length grows, has also been a subject of recent investigation \cite{MHU15unif-ieeeit, fong2017scaling, wang2018polar, blasiok2018polar}.

For practical block lengths, polar codes' error-correction performance under SC decoding is not satisfactory. Therefore, an SC list (SCL) decoder with time complexity $O\left(L N\log_2 N\right)$ and space complexity $O\left(L N\right)$ is used \cite{TVa15}, where $L$ is the size of the list. SCL decoding runs $L$ coupled SC decoders in parallel and maintains a list of the most likely codewords. The SCL decoder's empirical performance is close to that of the optimal MAP decoder with practical list-size $L$. Furthermore, by adding some extra bits of cyclic redundancy check (CRC) precoding, the performance is comparable to state-of-the-art low-density parity-check (LDPC) codes. 

SC-based decoding algorithms, such as SC and SCL decoding, suffer from high latency. This is due to the fact that SC decoding is inherently a serial algorithm: it proceeds sequentially bit by bit. In order to mitigate this issue, a \emph{simplified SC} (SSC) decoder was proposed in \cite{alamdar}. The SSC decoder identifies two specific constituent codes in the polar code, namely, constituent codes of rate $0$ (Rate-0) and rate $1$ (Rate-1). The bits within each constituent code can be decoded in parallel; thus, these constituent codes are decoded in one shot. Consequently, the latency is reduced without increasing the error probability. In \cite{sarkis, hanif, condo2018fast}, more constituent codes were identified and low-complexity parallel decoders were designed, increasing the throughput and reducing the latency even further. These results were extended to SCL decoders in \cite{hashemi_SSCL_TCASI, hashemi_FSSCL_TSP}. Recently, it was shown in \cite{mondelli2020sublinear} that the latency of the SSC decoder proposed in \cite{alamdar} is $O\left(N^{1-1/\mu}\right)$. Thus the latency of SSC decoding is \emph{sublinear in $N$}, in contrast to the $O\left(N\right)$ latency of standard SC decoding \cite{Ari09}. However, these results are based on the assumption that the hardware resources are unlimited, and thus a fully-parallel architecture can be implemented. In a practical application, this assumption is no longer valid and a specific number of processing elements (PEs) $P$ are allocated to perform the operations in SC-based decoding algorithms \cite{Leroux2012}. In the extreme case where $P=1$ (a fully-serial architecture), the latency of SC decoding grows from $O\left(N\right)$ to $O\left(N \log_2 N\right)$.

This paper quantifies the latency of the SSC decoder proposed in \cite{alamdar} as a function of hardware resource constraints. Our main result is that the latency of SSC decoding scales as
\begin{equation}
\label{main-bound}
O\left(N^{1-1/\mu}+\frac{N}{P}\log_2\log_2\frac{N}{P}\right)
\end{equation}
with the block length $N$. Several consequences of the bound in (\ref{main-bound}) are as follows. In a fully-parallel implementation, where $P=\frac{N}{2}$, this bound reduces to $O\left(N^{1-1/\mu}\right)$, thereby recovering the main result of \cite{mondelli2020sublinear}. In a fully-serial implementation, where $P=1$, the bound in (\ref{main-bound}) reduces to $O\left(N\log_2\log_2 N\right)$. This aligns with the results of \cite{wang2019log}, wherein a variant of polar codes with log-logarithmic complexity per information bit has been introduced. However, this paper's analysis is for \emph{conventional} polar codes rather than a variant thereof. Moreover, for the case where $P=1$, we determine the multiplicative constant in our bound and further refine it to $\left(2+o(1)\right)N\log_2\log_2 N$. Finally, it is shown that $P=N^{\frac{1}{\mu}}$ is the smallest number of processing elements that, asymptotically, provides the same latency as that of the fully-parallel decoder. The applicability of the foregoing results is validated through extensive simulations. Our numerical results confirm the presented bounds' tightness.

The rest of this paper is organized as follows: Section\,\ref{sec:prel} explains polar codes and discusses SC and SSC decoding algorithms with limited number of PEs; Section\,\ref{sec:main} states and proves that in an implementation of the SSC decoder with $P$ processing elements, the latency is upper bounded by $O\left(N^{1-1/\mu}+\frac{N}{P}\log_2\log_2\frac{N}{P}\right)$; numerical results are presented in Section\,\ref{sec:numerical} to verify the proposed bounds; and conclusions are drawn in Section\,\ref{sec:concl}.

%----------------------------------------------------
\section{Polar Coding Preliminaries} \label{sec:prel}
%----------------------------------------------------

%- - - - - - - - - - - - - - - - - - - - - - - - - - - -
\subsection{Polar Codes}
%- - - - - - - - - - - - - - - - - - - - - - - - - - - -

Consider a BMS channel $W: \mathcal{X}\to \mathcal{Y}$ defined by transition probabilities $\{W(y \mid x) : x\in \mathcal{X}, y\in \mathcal{Y}\}$, where $\mathcal{X}=\{0,1\}$ is the input alphabet and $\mathcal{Y}$ is an arbitrary output alphabet. The reliability of the channel $W$ can be measured by its Bhattacharyya parameter $Z(W)= \sum_{y \in \mathcal{Y}} \sqrt{W(y\mid 0)W(y \mid 1)}$. Channel polarization \cite{Ari09} is the process of mapping two copies of the channel $W$ into two synthetic channels $W^0: \mathcal{X}\to \mathcal{Y}^2$ and $W^1:\mathcal{X}\to \mathcal{X}\times\mathcal{Y}^2$ as
\begin{equation}\label{eq:mapch}
\begin{split}
W^0(y_1, y_2\mid x_1) & = \sum_{x_2\in \mathcal X} \frac{1}{2}W(y_1\mid x_1 \oplus x_2) W(y_2\mid x_2),\\
W^1(y_1, y_2, x_1\mid x_2) & = \frac{1}{2}W(y_1\mid x_1 \oplus x_2) W(y_2\mid x_2),\\
\end{split}
\end{equation}
where $W^0$ is a \emph{worse} channel and $W^1$ is a \emph{better} channel than $W$ because \cite{Ari09,RiU08}
\begin{align}
Z(W)\sqrt{2-Z(W)^2}&\le Z(W^0)\le 2Z(W)-Z(W)^2,\label{eq:minusB}\\
&Z(W^1)=Z(W)^2.\label{eq:plusB}
\end{align}
By recursively performing the operation in \eqref{eq:mapch} $n$ times, $2^n$ copies of $W$ are transformed into $2^n$ synthetic channels $W_n^{(i)} = (((W^{b_1^{(i)}})^{b_2^{(i)}})^{\cdots})^{b_n^{(i)}}$, where $1\leq i\leq 2^n$ and $(b_1^{(i)}, \ldots, b_n^{(i)})$ is the binary representation of the integer $i-1$ over $n$ bits. Consider a random sequence of channels, defined recursively as
\begin{equation}
W_{n} = \left\{ \begin{array}{ll}W_{n-1}^0, & \mbox{ w.p. } 1/2,\\ W_{n-1}^1,& \mbox{ w.p. } 1/2,\\ \end{array}\right.
\end{equation}
where $W_0=W$. Using \eqref{eq:minusB} and \eqref{eq:plusB}, the random process that tracks the Bhattacharyya parameter of $W_n$ can be represented as
\begin{equation}\label{eq:eqBMSC}
Z_{n} \left\{ \begin{array}{ll}\in \left[Z_{n-1}\sqrt{2-Z^2_{n-1}},\, 2 Z_{n-1}-Z^2_{n-1}\right], & \mbox{ w.p. } 1/2,\\ =Z^2_{n-1},& \mbox{ w.p. } 1/2,\\ \end{array}\right.
\end{equation}
where $Z_n=Z(W_n)$ and $n\ge 1$.

The construction of polar codes comprises the assigning of information bits to the set of positions with the best Bhattacharyya parameters, as stated in the following definition.

\begin{definition}[Polar code construction]\label{def:construction}
For a given block length $N=2^n$, BMS channel $W$, and probability of error $p_e\in (0, 1)$, the polar code $\mathcal C_{\rm polar}(p_e, W, N)$ is constructed by assigning the information bits to the positions corresponding to all the synthetic channels whose Bhattacharyya parameter is less than $p_e/N$ and by assigning a predefined (frozen) value to the remaining positions.
\end{definition}

With the construction rule of Definition~\ref{def:construction}, the error probability under SC decoding is guaranteed to be \emph{at most} $p_e$. Moreover, this construction rule ensures that the rate $R$ of the code tends to capacity at a speed that is captured by the \emph{scaling exponent} of the channel.

\begin{definition}[Upper bound on scaling exponent]\label{def:upscal}
We say that $\mu$ is an upper bound on the scaling exponent if there exists a function $h(x): [0, 1] \to [0, 1]$ such that $h(0)=h(1)=0$, $h(x)>0$ for any $x\in (0, 1)$, and
\begin{equation}\label{eq:suph}
\displaystyle\sup_{\substack{x\in (0, 1)\\y \in [x\sqrt{2-x^2}, 2x-x^2]}}\displaystyle\frac{h(x^2)+h(y)}{2h(x)} < 2^{-1/\mu}.
\end{equation}
\end{definition}

By defining the scaling exponent as in Definition~\ref{def:upscal}, the gap to capacity $I(W)-R$ scales as $O(N^{-1/\mu})$, see Theorem~1 of \cite{MHU15unif-ieeeit}. Note that $\mu \approx 4$ for BAWGNC as shown in \cite{KMTU10}, and it is conjectured that $\mu \approx 4.2$ for BSC. For the BEC, the condition \eqref{eq:suph} can be relaxed to
\begin{equation}
\displaystyle\sup_{x\in (0, 1)}\displaystyle\frac{h(x^2)+h(2x-x^2)}{2h(x)} < 2^{-1/\mu},
\end{equation}
which gives a numerical value $\mu\approx 3.63$.

%- - - - - - - - - - - - - - - - - - - - - - - - - - - -
\subsection{Successive-Cancellation Decoding}
%- - - - - - - - - - - - - - - - - - - - - - - - - - - -

SC decoding is a message passing algorithm on the factor graph of polar codes, as shown in Fig.~\ref{fig:pcEncDec} for a polar code of length $N=8$. At stage $n$ of the factor graph, the LLR values $\bm{\alpha}_n^{0:N-1} = \{\alpha_n^0,\alpha_n^1,\ldots,\alpha_n^{N-1}\}$, that are calculated from the received channel-output vector, are fed to the decoder. Fig.~\ref{fig:pcDec8} shows how the vector of internal LLR values, $\bm{\alpha}_s^{0:N-1} = \{\alpha_s^0,\alpha_s^1,\ldots,\alpha_s^{N-1}\}$, which is composed of $\frac{N}{2^s}$ vectors of $2^s$ LLR values $\bm{\alpha}_s^{i2^s:(i+1)2^s-1} = \{\alpha_s^{i2^s},\alpha_s^{i2^s+1},\ldots,\alpha_s^{(i+1)2^s-1}\}$, is generated. Specifically, at each level $s$, we have:
\begin{equation}
\alpha_s^i =
    \begin{cases}
    f(\alpha_{s+1}^i,\alpha_{s+1}^{i+2^s}) & \text{if $\floor{\frac{i}{2^s}}\, \mathrm{mod}\,2 = 0$,} \\
    g(\alpha_{s+1}^{i},\alpha_{s+1}^{i-2^s},\beta_{s}^{i-2^s}) & \text{if $\floor{\frac{i}{2^s}}\, \mathrm{mod}\,2 = 1$,}
    \end{cases} \label{eq:alpha}
\end{equation}
where $f(a,b) = 2\arctanh \left(\tanh\left(\frac{a}{2}\right)\tanh\left(\frac{b}{2}\right)\right)$, $g(a,b,c) = a + (1-2c)b$, and $\beta_s^i$ is the $i$-th bit estimate at level $s$ of the factor graph. As shown in Fig.~\ref{fig:pcEnc8}, the bit estimates $\bm{\beta}_s = \{\beta_s^0,\beta_s^1,\ldots,\beta_s^{N-1}\}$ are calculated as
\begin{equation}
\beta_s^i =
    \begin{cases}
    \beta_{s-1}^{i} \oplus \beta_{s-1}^{i+2^s} & \text{if $\floor{\frac{i}{2^s}}\, \mathrm{mod}\,2 = 0$,} \\
    \beta_{s-1}^{i} & \text{if $\floor{\frac{i}{2^s}}\, \mathrm{mod}\,2 = 1$,}
    \end{cases} \label{eq:beta}
\end{equation}
where $\oplus$ is the bit-wise XOR operation. All frozen bits are assumed to be zero. Hence at level $s=0$, the $i$-th bit $\hat{u}_i$ is estimated as
\begin{equation}
    \hat{u}_i = \beta_{0}^{i} =
    \begin{cases}
    0 & \text{if $u_i$ is a frozen bit or $\alpha_0^i>0$,} \\
    1 & \text{otherwise.}
    \end{cases}
\end{equation}

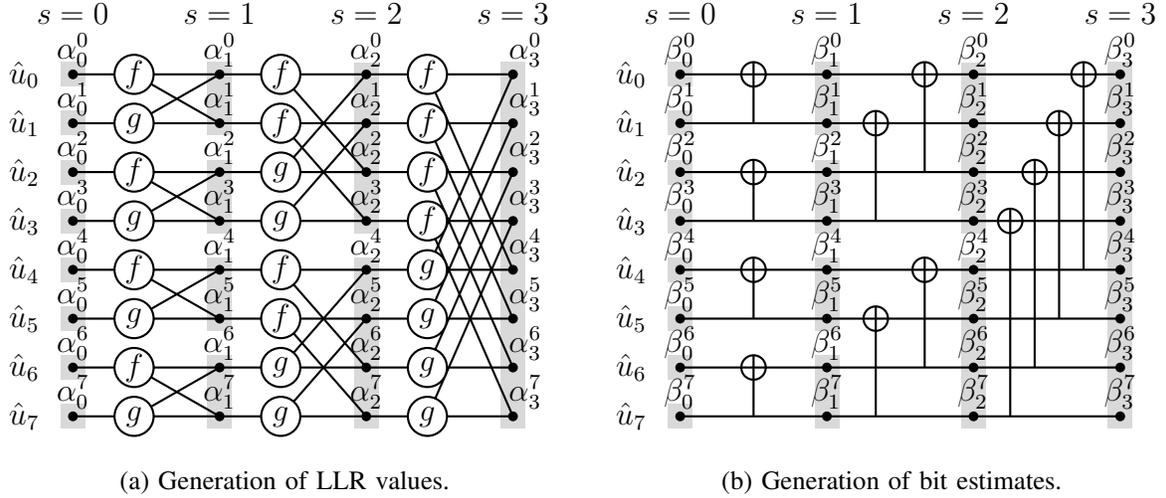
\begin{figure*}[t]
\centering
\begin{subfigure}{.48\textwidth}
\centering
\begin{tikzpicture}[scale=.65, thick]

  \node at (0,0) {$\hat{u}_0$};
  \node at (0,-1) {$\hat{u}_1$};
  \node at (0,-2) {$\hat{u}_2$};
  \node at (0,-3) {$\hat{u}_3$};
  \node at (0,-4) {$\hat{u}_4$};
  \node at (0,-5) {$\hat{u}_5$};
  \node at (0,-6) {$\hat{u}_6$};
  \node at (0,-7) {$\hat{u}_7$};

  \foreach \x in {-7,...,0}
  {
    \fill[color=black!15] (.75,\x+.25) rectangle ++(.5,-.5);
  }
  
  \foreach \x in {-6,-4,-2,0}
  {
    \fill[color=black!15] (3.75,\x+.25) rectangle ++(.5,-1.5);
  }

  \foreach \x in {-4,0}
  {
    \fill[color=black!15] (6.75,\x+.25) rectangle ++(.5,-3.5);
  }
 
  \fill[color=black!15] (9.75,.25) rectangle ++(.5,-7.5);

  \foreach \x in {-7,...,0}
  {
    \draw (1,\x) -- (10,\x);
    \fill (1,\x) circle [radius=.1];
    \fill (4,\x) circle [radius=.1];
    \fill (7,\x) circle [radius=.1];
    \fill (10,\x) circle [radius=.1];
  }
  
  \node at (1,1.25) {$s=0$};
  \node at (4,1.25) {$s=1$};
  \node at (7,1.25) {$s=2$};
  \node at (10,1.25) {$s=3$};

  \node at (1,.5) {$\alpha_{0}^{0}$};
  \node at (1,-.5) {$\alpha_{0}^{1}$};
  \node at (1,-1.5) {$\alpha_{0}^{2}$};
  \node at (1,-2.5) {$\alpha_{0}^{3}$};
  \node at (1,-3.5) {$\alpha_{0}^{4}$};
  \node at (1,-4.5) {$\alpha_{0}^{5}$};
  \node at (1,-5.5) {$\alpha_{0}^{6}$};
  \node at (1,-6.5) {$\alpha_{0}^{7}$};

  \node at (4,.5) {$\alpha_{1}^{0}$};
  \node at (4,-.5) {$\alpha_{1}^{1}$};
  \node at (4,-1.5) {$\alpha_{1}^{2}$};
  \node at (4,-2.5) {$\alpha_{1}^{3}$};
  \node at (4,-3.5) {$\alpha_{1}^{4}$};
  \node at (4,-4.5) {$\alpha_{1}^{5}$};
  \node at (4,-5.5) {$\alpha_{1}^{6}$};
  \node at (4,-6.5) {$\alpha_{1}^{7}$};

  \node at (7,.5) {$\alpha_{2}^{0}$};
  \node at (7,-.5) {$\alpha_{2}^{1}$};
  \node at (7,-1.5) {$\alpha_{2}^{2}$};
  \node at (7,-2.5) {$\alpha_{2}^{3}$};
  \node at (7,-3.5) {$\alpha_{2}^{4}$};
  \node at (7,-4.5) {$\alpha_{2}^{5}$};
  \node at (7,-5.5) {$\alpha_{2}^{6}$};
  \node at (7,-6.5) {$\alpha_{2}^{7}$};

  \node at (10.25,.5) {$\alpha_{3}^{0}$};
  \node at (10.25,-.5) {$\alpha_{3}^{1}$};
  \node at (10.25,-1.5) {$\alpha_{3}^{2}$};
  \node at (10.25,-2.5) {$\alpha_{3}^{3}$};
  \node at (10.25,-3.5) {$\alpha_{3}^{4}$};
  \node at (10.25,-4.5) {$\alpha_{3}^{5}$};
  \node at (10.25,-5.5) {$\alpha_{3}^{6}$};
  \node at (10.25,-6.5) {$\alpha_{3}^{7}$};
  
  \draw (2.25,0) -- (4,-1);
  \draw (4,0) -- (2.25,-1);
  \draw (2.25,-2) -- (4,-3);
  \draw (4,-2) -- (2.25,-3);
  \draw (2.25,-4) -- (4,-5);
  \draw (4,-4) -- (2.25,-5);
  \draw (2.25,-6) -- (4,-7);
  \draw (4,-6) -- (2.25,-7);

  \draw (5.25,0) -- (7,-2);
  \draw (7,0) -- (5.25,-2);
  \draw (5.25,-1) -- (7,-3);
  \draw (7,-1) -- (5.25,-3);
  \draw (5.25,-4) -- (7,-6);
  \draw (7,-4) -- (5.25,-6);
  \draw (5.25,-5) -- (7,-7);
  \draw (7,-5) -- (5.25,-7);

  \draw (8.25,0) -- (10,-4);
  \draw (10,0) -- (8.25,-4);
  \draw (8.25,-1) -- (10,-5);
  \draw (10,-1) -- (8.25,-5);
  \draw (8.25,-2) -- (10,-6);
  \draw (10,-2) -- (8.25,-6);
  \draw (8.25,-3) -- (10,-7);
  \draw (10,-3) -- (8.25,-7);

  \foreach \x in {-7,...,0}
  {
    \foreach \y in {2.25,5.25,8.25}
    {
      \draw[fill=white] (\y,\x) circle [radius=.4];
    }
  }

  \node at (2.25,0) {$f$};
  \node at (2.25,-1) {$g$};
  \node at (2.25,-2) {$f$};
  \node at (2.25,-3) {$g$};
  \node at (2.25,-4) {$f$};
  \node at (2.25,-5) {$g$};
  \node at (2.25,-6) {$f$};
  \node at (2.25,-7) {$g$};

  \node at (5.25,0) {$f$};
  \node at (5.25,-1) {$f$};
  \node at (5.25,-2) {$g$};
  \node at (5.25,-3) {$g$};
  \node at (5.25,-4) {$f$};
  \node at (5.25,-5) {$f$};
  \node at (5.25,-6) {$g$};
  \node at (5.25,-7) {$g$};

  \node at (8.25,0) {$f$};
  \node at (8.25,-1) {$f$};
  \node at (8.25,-2) {$f$};
  \node at (8.25,-3) {$f$};
  \node at (8.25,-4) {$g$};
  \node at (8.25,-5) {$g$};
  \node at (8.25,-6) {$g$};
  \node at (8.25,-7) {$g$};
  
\end{tikzpicture}
\caption{Generation of LLR values.}
\label{fig:pcDec8}
\end{subfigure}
\begin{subfigure}{.48\textwidth}
\centering
\begin{tikzpicture}[scale=.65, thick]

  \node at (0,0) {$\hat{u}_0$};
  \node at (0,-1) {$\hat{u}_1$};
  \node at (0,-2) {$\hat{u}_2$};
  \node at (0,-3) {$\hat{u}_3$};
  \node at (0,-4) {$\hat{u}_4$};
  \node at (0,-5) {$\hat{u}_5$};
  \node at (0,-6) {$\hat{u}_6$};
  \node at (0,-7) {$\hat{u}_7$};

  \foreach \x in {-7,...,0}
  {
    \fill[color=black!15] (.75,\x+.25) rectangle ++(.5,-.5);
  }
  
  \foreach \x in {-6,-4,-2,0}
  {
    \fill[color=black!15] (3.75,\x+.25) rectangle ++(.5,-1.5);
  }

  \foreach \x in {-4,0}
  {
    \fill[color=black!15] (6.75,\x+.25) rectangle ++(.5,-3.5);
  }
 
  \fill[color=black!15] (9.75,.25) rectangle ++(.5,-7.5);

  \foreach \x in {-7,...,0}
  {
    \draw (1,\x) -- (10,\x);
    \fill (1,\x) circle [radius=.1];
    \fill (4,\x) circle [radius=.1];
    \fill (7,\x) circle [radius=.1];
    \fill (10,\x) circle [radius=.1];
  }
  
  \node at (1,1.25) {$s=0$};
  \node at (4,1.25) {$s=1$};
  \node at (7,1.25) {$s=2$};
  \node at (10,1.25) {$s=3$};

  \node at (1,.5) {$\beta_{0}^{0}$};
  \node at (1,-.5) {$\beta_{0}^{1}$};
  \node at (1,-1.5) {$\beta_{0}^{2}$};
  \node at (1,-2.5) {$\beta_{0}^{3}$};
  \node at (1,-3.5) {$\beta_{0}^{4}$};
  \node at (1,-4.5) {$\beta_{0}^{5}$};
  \node at (1,-5.5) {$\beta_{0}^{6}$};
  \node at (1,-6.5) {$\beta_{0}^{7}$};

  \node at (4,.5) {$\beta_{1}^{0}$};
  \node at (4,-.5) {$\beta_{1}^{1}$};
  \node at (4,-1.5) {$\beta_{1}^{2}$};
  \node at (4,-2.5) {$\beta_{1}^{3}$};
  \node at (4,-3.5) {$\beta_{1}^{4}$};
  \node at (4,-4.5) {$\beta_{1}^{5}$};
  \node at (4,-5.5) {$\beta_{1}^{6}$};
  \node at (4,-6.5) {$\beta_{1}^{7}$};

  \node at (7,.5) {$\beta_{2}^{0}$};
  \node at (7,-.5) {$\beta_{2}^{1}$};
  \node at (7,-1.5) {$\beta_{2}^{2}$};
  \node at (7,-2.5) {$\beta_{2}^{3}$};
  \node at (7,-3.5) {$\beta_{2}^{4}$};
  \node at (7,-4.5) {$\beta_{2}^{5}$};
  \node at (7,-5.5) {$\beta_{2}^{6}$};
  \node at (7,-6.5) {$\beta_{2}^{7}$};

  \node at (10,.5) {$\beta_{3}^{0}$};
  \node at (10,-.5) {$\beta_{3}^{1}$};
  \node at (10,-1.5) {$\beta_{3}^{2}$};
  \node at (10,-2.5) {$\beta_{3}^{3}$};
  \node at (10,-3.5) {$\beta_{3}^{4}$};
  \node at (10,-4.5) {$\beta_{3}^{5}$};
  \node at (10,-5.5) {$\beta_{3}^{6}$};
  \node at (10,-6.5) {$\beta_{3}^{7}$};
  
  \foreach \x in {-7,-5,-3,-1}
  {
    \draw (2.5,\x) -- (2.5,\x+1.25);
    \draw (2.5,\x+1) circle [radius=.25];
  }
  
  \draw (6,-2) -- (6,.25);
  \draw (5,-3) -- (5,-.75);
  \draw (6,0) circle [radius=.25];
  \draw (5,-1) circle [radius=.25];
  \draw (6,-6) -- (6,-3.75);
  \draw (5,-7) -- (5,-4.75);
  \draw (6,-4) circle [radius=.25];
  \draw (5,-5) circle [radius=.25];

  \draw (9.25,-4) -- (9.25,.25);
  \draw (8.75,-5) -- (8.75,-.75);
  \draw (8.25,-6) -- (8.25,-1.75);
  \draw (7.75,-7) -- (7.75,-2.75);
  \draw (9.25,0) circle [radius=.25];
  \draw (8.75,-1) circle [radius=.25];
  \draw (8.25,-2) circle [radius=.25];
  \draw (7.75,-3) circle [radius=.25];

\end{tikzpicture}
\caption{Generation of bit estimates.}
\label{fig:pcEnc8}
\end{subfigure}
\caption{SC decoding on the factor graph representation of polar codes with $N=8$. Each gray area represents one node in the binary tree representation of SC decoding.}
\label{fig:pcEncDec}
\end{figure*}

By combining all the operations in (\ref{eq:alpha}) that can be performed in parallel, SC decoding can be represented as on Fig.~\ref{fig:scDec}'s binary tree. Fig.~\ref{fig:scDec}'s root node at decoding stage $n$ is fed with the LLR values, and the results of operations in (\ref{eq:alpha}) and (\ref{eq:beta}) are passed on the branches of the decoding tree. SC decoding has a sequential structure in the sense that the decoding of each bit depends on the decoding of its previous bits. More formally, on the one hand, when $\mod(\frac{i}{2^s},2) = 0$, the calculation of $\alpha_s^i$ at level $s$ is only dependent on the LLR values that are received from a node at level $s-1$. On the other hand, when $\mod(\frac{i}{2^s},2) = 1$, the calculation of $\alpha_s^i$ also depends on a hard bit estimation $\beta^{i-2^s}_s$ that is a result of estimating the previous bits (see (\ref{eq:alpha})). Consequently, SC decoding proceeds by traversing the binary tree such that the nodes at level $s=0$ are visited from left to right.

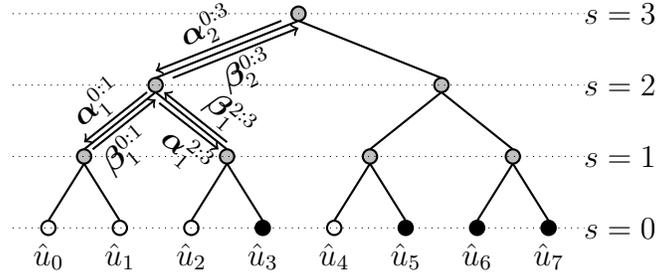
\begin{figure}[t]
\centering
\begin{tikzpicture}[scale=1.9, thick]
  \draw [fill=lightgray] (0,0) circle [radius=.05];

  \draw [fill=lightgray] (-1,-.5) circle [radius=.05];
  \draw [fill=lightgray] (1,-.5) circle [radius=.05];

  \draw [fill=lightgray] (-1.5,-1) circle [radius=.05];
  \draw [fill=lightgray] (-.5,-1) circle [radius=.05];
  \draw [fill=lightgray] (.5,-1) circle [radius=.05];
  \draw [fill=lightgray] (1.5,-1) circle [radius=.05];

  \draw (-1.75,-1.5) circle [radius=.05];
  \draw (-1.25,-1.5) circle [radius=.05];
  \draw (-.75,-1.5) circle [radius=.05];
  \draw [fill=black] (-.25,-1.5) circle [radius=.05];
  \draw (.25,-1.5) circle [radius=.05];
  \draw [fill=black] (.75,-1.5) circle [radius=.05];
  \draw [fill=black] (1.25,-1.5) circle [radius=.05];
  \draw [fill=black] (1.75,-1.5) circle [radius=.05];

  \node at (-1.75,-1.7) {$\hat{u}_0$};
  \node at (-1.25,-1.7) {$\hat{u}_1$};
  \node at (-.75,-1.7) {$\hat{u}_2$};
  \node at (-.25,-1.7) {$\hat{u}_3$};
  \node at (.25,-1.7) {$\hat{u}_4$};
  \node at (.75,-1.7) {$\hat{u}_5$};
  \node at (1.25,-1.7) {$\hat{u}_6$};
  \node at (1.75,-1.7) {$\hat{u}_7$};

  \draw (0,-.05) -- (-1,-.45);
  \draw (0,-.05) -- (1,-.45);

  \draw (-1,-.55) -- (-1.5,-.95);
  \draw (-1,-.55) -- (-.5,-.95);
  \draw (1,-.55) -- (.5,-.95);
  \draw (1,-.55) -- (1.5,-.95);

  \draw (-1.5,-1.05) -- (-1.75,-1.45);
  \draw (-1.5,-1.05) -- (-1.25,-1.45);
  \draw (-.5,-1.05) -- (-.75,-1.45);
  \draw (-.5,-1.05) -- (-.25,-1.45);
  \draw (.5,-1.05) -- (.25,-1.45);
  \draw (.5,-1.05) -- (.75,-1.45);
  \draw (1.5,-1.05) -- (1.25,-1.45);
  \draw (1.5,-1.05) -- (1.75,-1.45);

  \draw [thin,dotted] (-2,0) -- (2,0);
  \draw [thin,dotted] (-2,-.5) -- (2,-.5);
  \draw [thin,dotted] (-2,-1) -- (2,-1);
  \draw [thin,dotted] (-2,-1.5) -- (2,-1.5);

  \node at (2.25,0) {$s=3$};
  \node at (2.25,-.5) {$s=2$};
  \node at (2.25,-1) {$s=1$};
  \node at (2.25,-1.5) {$s=0$};

  \draw [->] (-.12,-.05) -- (-1,-.4) node [above=-.1cm,midway,rotate=25] {$\bm{\alpha}_2^{0:3}$};
  \draw [->] (-.88,-.45) -- (0,-.1) node [below=-.1cm,midway,rotate=25] {$\bm{\beta}_2^{0:3}$};

  \draw [->] (-1.06,-.55) -- (-1.5,-.9) node [above=-.1cm,midway,rotate=40] {$\bm{\alpha}_1^{0:1}$};
  \draw [->] (-1.44,-.95) -- (-1.0,-0.6) node [below=-.1cm,near start,rotate=40] {$\bm{\beta}_1^{0:1}$};

  \draw [<-] (-.94,-.55) -- (-.5,-.9) node [above=-.1cm,near end,rotate=-40] {$\bm{\beta}_1^{2:3}$};
  \draw [<-] (-.56,-.95) -- (-0.975,-.625) node [below=-.1cm,near start,rotate=-40] {$\bm{\alpha}_1^{2:3}$};

\end{tikzpicture}
\caption{Binary tree representation of SC decoding for a polar code with $N=8$ and $R=1/2$. The white nodes represent frozen bits and the black nodes represent information bits.}
\label{fig:scDec}
\end{figure}

All operations at a specific SC-decoding-tree node can be in principle performed in parallel. However, when the SC-decoder hardware implementation is considered, the number of PEs that perform the calculations in (\ref{eq:alpha}) is constrained to a specific value $P$, which can improve the trade-off between chip area and latency \cite{Leroux2012}. As shown in \cite{Leroux2012}, if the channel LLR values are readily available, then the latency of SC decoding is
\begin{equation}
    \mathcal{L} = 2N + \frac{N}{P}\log_2\left(\frac{N}{4P}\right) \text{.}
\end{equation}

For different values of $P$, Fig.~\ref{fig:pcSch} shows the resulting LLR values at each time step in a length $N=8$ polar code. When $P=\frac{N}{2}$, the decoder can perform all the parallelizable operations in one time step, thus the implementation is \emph{fully-parallel} (see Fig.~\ref{fig:pcSchFP}). When $P=1$, only one operation can be performed at each time step, thus the implementation is \emph{fully-serial} (see Fig.~\ref{fig:pcSchFS}). Any $P$ in the interval $(1,\frac{N}{2})$ results in a \emph{semi-parallel} implementation (see Fig.~\ref{fig:pcSchSP}).

\begin{figure}[t]
\centering
\begin{subfigure}{\textwidth}
\centering
\begin{tikzpicture}[scale=.6]

\node at (-1.5,1) {time};
\node at (-1.5,0) {$\text{PE}_1$};
\node at (-1.5,-1) {$\text{PE}_2$};
\node at (-1.5,-2) {$\text{PE}_3$};
\node at (-1.5,-3) {$\text{PE}_4$};
\node at (-1.5,-4) {output};

\node at (0,0) {$\alpha_2^0$};
\node at (0,-1) {$\alpha_2^1$};
\node at (0,-2) {$\alpha_2^2$};
\node at (0,-3) {$\alpha_2^3$};
\node at (7,0) {$\alpha_2^4$};
\node at (7,-1) {$\alpha_2^5$};
\node at (7,-2) {$\alpha_2^6$};
\node at (7,-3) {$\alpha_2^7$};

\node at (1,0) {$\alpha_1^0$};
\node at (1,-1) {$\alpha_1^1$};
\node at (7+1,0) {$\alpha_1^4$};
\node at (7+1,-1) {$\alpha_1^5$};

\node at (2,0) {$\alpha_0^0$};
\node at (2+7,0) {$\alpha_0^4$};

\node at (3,0) {$\alpha_0^1$};
\node at (3+7,0) {$\alpha_0^5$};

\node at (4,0) {$\alpha_1^2$};
\node at (4,-1) {$\alpha_1^3$};
\node at (7+4,0) {$\alpha_1^6$};
\node at (7+4,-1) {$\alpha_1^7$};

\node at (5,0) {$\alpha_0^2$};
\node at (5+7,0) {$\alpha_0^6$};

\node at (6,0) {$\alpha_0^3$};
\node at (6+7,0) {$\alpha_0^7$};

\foreach \x in {0,...,13}
{
\draw [very thin] (\x+.5,1.5) -- (\x+.5,-4.5);
}

\foreach \x in {-1}
{
\draw [very thick] (\x+.5,1.5) -- (\x+.5,-4.5);
}

\foreach \x in {-3,...,-1}
{
\draw [very thin] (-2.5,\x+.5) -- (13.5,\x+.5);
}

\foreach \x in {-4,0}
{
\draw [very thick] (-2.5,\x+.5) -- (13.5,\x+.5);
}

\node at (0,1) {$1$};
\node at (1,1) {$2$};
\node at (2,1) {$3$};
\node at (3,1) {$4$};
\node at (4,1) {$5$};
\node at (5,1) {$6$};
\node at (6,1) {$7$};
\node at (7,1) {$8$};
\node at (8,1) {$9$};
\node at (9,1) {$10$};
\node at (10,1) {$11$};
\node at (11,1) {$12$};
\node at (12,1) {$13$};
\node at (13,1) {$14$};

\node at (2,-4) {$\hat{u}_0$};
\node at (3,-4) {$\hat{u}_1$};
\node at (5,-4) {$\hat{u}_2$};
\node at (6,-4) {$\hat{u}_3$};
\node at (9,-4) {$\hat{u}_4$};
\node at (10,-4) {$\hat{u}_5$};
\node at (12,-4) {$\hat{u}_6$};
\node at (13,-4) {$\hat{u}_7$};

\end{tikzpicture}
\caption{Fully-parallel ($P=4$).}
\label{fig:pcSchFP}
\end{subfigure}
\begin{subfigure}{\textwidth}
\centering
\begin{tikzpicture}[scale=.6]

\node at (-1.5,1) {time};
\node at (-1.5,0) {$\text{PE}_1$};
\node at (-1.5,-1) {$\text{PE}_2$};
\node at (-1.5,-2) {output};

\node at (0,0) {$\alpha_2^0$};
\node at (0,-1) {$\alpha_2^1$};
\node at (1,0) {$\alpha_2^2$};
\node at (1,-1) {$\alpha_2^3$};

\node at (2,0) {$\alpha_1^0$};
\node at (2,-1) {$\alpha_1^1$};

\node at (3,0) {$\alpha_0^0$};

\node at (4,0) {$\alpha_0^1$};

\node at (5,0) {$\alpha_1^2$};
\node at (5,-1) {$\alpha_1^3$};

\node at (6,0) {$\alpha_0^2$};

\node at (7,0) {$\alpha_0^3$};

\node at (8,0) {$\alpha_2^4$};
\node at (8,-1) {$\alpha_2^5$};
\node at (9,0) {$\alpha_2^6$};
\node at (9,-1) {$\alpha_2^7$};

\node at (10,0) {$\alpha_1^4$};
\node at (10,-1) {$\alpha_1^5$};

\node at (11,0) {$\alpha_0^4$};

\node at (12,0) {$\alpha_0^5$};

\node at (13,0) {$\alpha_1^6$};
\node at (13,-1) {$\alpha_1^7$};

\node at (14,0) {$\alpha_0^6$};

\node at (15,0) {$\alpha_0^7$};

\foreach \x in {0,...,15}
{
\draw [very thin] (\x+.5,1.5) -- (\x+.5,-2.5);
}

\foreach \x in {-1}
{
\draw [very thick] (\x+.5,1.5) -- (\x+.5,-2.5);
}

\foreach \x in {-1}
{
\draw [very thin] (-2.5,\x+.5) -- (15.5,\x+.5);
}

\foreach \x in {-2,0}
{
\draw [very thick] (-2.5,\x+.5) -- (15.5,\x+.5);
}

\node at (0,1) {$1$};
\node at (1,1) {$2$};
\node at (2,1) {$3$};
\node at (3,1) {$4$};
\node at (4,1) {$5$};
\node at (5,1) {$6$};
\node at (6,1) {$7$};
\node at (7,1) {$8$};
\node at (8,1) {$9$};
\node at (9,1) {$10$};
\node at (10,1) {$11$};
\node at (11,1) {$12$};
\node at (12,1) {$13$};
\node at (13,1) {$14$};
\node at (14,1) {$15$};
\node at (15,1) {$16$};

\node at (3,-2) {$\hat{u}_0$};
\node at (4,-2) {$\hat{u}_1$};
\node at (6,-2) {$\hat{u}_2$};
\node at (7,-2) {$\hat{u}_3$};
\node at (11,-2) {$\hat{u}_4$};
\node at (12,-2) {$\hat{u}_5$};
\node at (14,-2) {$\hat{u}_6$};
\node at (15,-2) {$\hat{u}_7$};

\end{tikzpicture}
\caption{Semi-parallel ($P=2$).}
\label{fig:pcSchSP}
\end{subfigure}
\begin{subfigure}{\textwidth}
\centering
\begin{tikzpicture}[scale=.6]

\node at (-1.5,1) {time};
\node at (-1.5,0) {$\text{PE}_1$};
\node at (-1.5,-1) {output};

\node at (0,0) {$\alpha_2^0$};
\node at (1,0) {$\alpha_2^1$};
\node at (2,0) {$\alpha_2^2$};
\node at (3,0) {$\alpha_2^3$};

\node at (4,0) {$\alpha_1^0$};
\node at (5,0) {$\alpha_1^1$};

\node at (6,0) {$\alpha_0^0$};

\node at (7,0) {$\alpha_0^1$};

\node at (8,0) {$\alpha_1^2$};
\node at (9,0) {$\alpha_1^3$};

\node at (10,0) {$\alpha_0^2$};

\node at (11,0) {$\alpha_0^3$};

\node at (12,0) {$\alpha_2^4$};
\node at (13,0) {$\alpha_2^5$};
\node at (14,0) {$\alpha_2^6$};
\node at (15,0) {$\alpha_2^7$};

\node at (16,0) {$\alpha_1^4$};
\node at (17,0) {$\alpha_1^5$};

\node at (18,0) {$\alpha_0^4$};

\node at (19,0) {$\alpha_0^5$};

\node at (20,0) {$\alpha_1^6$};
\node at (21,0) {$\alpha_1^7$};

\node at (22,0) {$\alpha_0^6$};

\node at (23,0) {$\alpha_0^7$};

\foreach \x in {0,...,23}
{
\draw [very thin] (\x+.5,1.5) -- (\x+.5,-1.5);
}

\foreach \x in {-1}
{
\draw [very thick] (\x+.5,1.5) -- (\x+.5,-1.5);
}

\foreach \x in {-1,0}
{
\draw [very thick] (-2.5,\x+.5) -- (23.5,\x+.5);
}

\node at (0,1) {$1$};
\node at (1,1) {$2$};
\node at (2,1) {$3$};
\node at (3,1) {$4$};
\node at (4,1) {$5$};
\node at (5,1) {$6$};
\node at (6,1) {$7$};
\node at (7,1) {$8$};
\node at (8,1) {$9$};
\node at (9,1) {$10$};
\node at (10,1) {$11$};
\node at (11,1) {$12$};
\node at (12,1) {$13$};
\node at (13,1) {$14$};
\node at (14,1) {$15$};
\node at (15,1) {$16$};
\node at (16,1) {$17$};
\node at (17,1) {$18$};
\node at (18,1) {$19$};
\node at (19,1) {$20$};
\node at (20,1) {$21$};
\node at (21,1) {$22$};
\node at (22,1) {$23$};
\node at (23,1) {$24$};

\node at (6,-1) {$\hat{u}_0$};
\node at (7,-1) {$\hat{u}_1$};
\node at (10,-1) {$\hat{u}_2$};
\node at (11,-1) {$\hat{u}_3$};
\node at (18,-1) {$\hat{u}_4$};
\node at (19,-1) {$\hat{u}_5$};
\node at (22,-1) {$\hat{u}_6$};
\node at (23,-1) {$\hat{u}_7$};

\end{tikzpicture}
\caption{Fully-serial ($P=1$).}
\label{fig:pcSchFS}
\end{subfigure}
\caption{SC decoding schedule for a polar code with $N=8$.}
\label{fig:pcSch}
\end{figure}

The latency of SC decoding can be represented on a binary tree by assigning \emph{decoding weights} to each edge based on the value of $P$, as illustrated in Fig.~\ref{fig:scDecLatency}. At each edge of the decoding tree that connects a node at level $s+1$ to a node at level $s$, the decoding weight is calculated as $\ceil{\frac{2^{s}}{P}}$, where $P$ is assumed to be a positive integer. In Fig.~\ref{fig:sc-decFP}'s fully-parallel implementation, all the edges have a decoding weight of $1$ since all the parallelizable operations are performed in parallel. However, in a fully-serial implementation of Fig.~\ref{fig:sc-decFS}, the edges at the top of the SC decoding tree consume more time steps, thus their decoding weights are larger. Using the binary tree representation, the latency of SC decoding can be calculated by adding the decoding weights on all the edges. Note that in a fully-parallel implementation, $\mathcal{L} = 2N-2$, and in a fully-serial implementation, $\mathcal{L} = N\log_2 N$. The latency in a fully-serial implementation is also the decoding complexity.

\begin{figure}[t]
\centering
\begin{subfigure}{.32\textwidth}
\centering
\begin{tikzpicture}[scale=1.2, thick]
  \draw [fill=lightgray] (0,0) circle [radius=.05];

  \draw [fill=lightgray] (-1,-.5) circle [radius=.05];
  \draw [fill=lightgray] (1,-.5) circle [radius=.05];

  \draw [fill=lightgray] (-1.5,-1) circle [radius=.05];
  \draw [fill=lightgray] (-.5,-1) circle [radius=.05];
  \draw [fill=lightgray] (.5,-1) circle [radius=.05];
  \draw [fill=lightgray] (1.5,-1) circle [radius=.05];

  \draw (-1.75,-1.5) circle [radius=.05];
  \draw (-1.25,-1.5) circle [radius=.05];
  \draw (-.75,-1.5) circle [radius=.05];
  \draw [fill=black] (-.25,-1.5) circle [radius=.05];
  \draw (.25,-1.5) circle [radius=.05];
  \draw [fill=black] (.75,-1.5) circle [radius=.05];
  \draw [fill=black] (1.25,-1.5) circle [radius=.05];
  \draw [fill=black] (1.75,-1.5) circle [radius=.05];

  \node at (-1.75,-1.7) {$\hat{u}_0$};
  \node at (-1.25,-1.7) {$\hat{u}_1$};
  \node at (-.75,-1.7) {$\hat{u}_2$};
  \node at (-.25,-1.7) {$\hat{u}_3$};
  \node at (.25,-1.7) {$\hat{u}_4$};
  \node at (.75,-1.7) {$\hat{u}_5$};
  \node at (1.25,-1.7) {$\hat{u}_6$};
  \node at (1.75,-1.7) {$\hat{u}_7$};

  \draw (0,-.05) -- (-1,-.45) node[pos=0.5,sloped,above] {$1$};
  \draw (0,-.05) -- (1,-.45) node[pos=0.5,sloped,above] {$1$};

  \draw (-1,-.55) -- (-1.5,-.95) node[pos=0.5,sloped,above] {$1$};
  \draw (-1,-.55) -- (-.5,-.95) node[pos=0.5,sloped,above] {$1$};
  \draw (1,-.55) -- (.5,-.95) node[pos=0.5,sloped,above] {$1$};
  \draw (1,-.55) -- (1.5,-.95) node[pos=0.5,sloped,above] {$1$};

  \draw (-1.5,-1.05) -- (-1.75,-1.45) node[pos=0.5,sloped,above] {$1$};
  \draw (-1.5,-1.05) -- (-1.25,-1.45) node[pos=0.5,sloped,above] {$1$};
  \draw (-.5,-1.05) -- (-.75,-1.45) node[pos=0.5,sloped,above] {$1$};
  \draw (-.5,-1.05) -- (-.25,-1.45) node[pos=0.5,sloped,above] {$1$};
  \draw (.5,-1.05) -- (.25,-1.45) node[pos=0.5,sloped,above] {$1$};
  \draw (.5,-1.05) -- (.75,-1.45) node[pos=0.5,sloped,above] {$1$};
  \draw (1.5,-1.05) -- (1.25,-1.45) node[pos=0.5,sloped,above] {$1$};
  \draw (1.5,-1.05) -- (1.75,-1.45) node[pos=0.5,sloped,above] {$1$};

  \draw [thin,dotted] (-2,0) -- (2,0);
  \draw [thin,dotted] (-2,-.5) -- (2,-.5);
  \draw [thin,dotted] (-2,-1) -- (2,-1);
  \draw [thin,dotted] (-2,-1.5) -- (2,-1.5);

\end{tikzpicture}
\caption{Fully-parallel ($P=4$).}
\label{fig:sc-decFP}
\end{subfigure}
\begin{subfigure}{.32\textwidth}
\centering
\begin{tikzpicture}[scale=1.2, thick]
  \draw [fill=lightgray] (0,0) circle [radius=.05];

  \draw [fill=lightgray] (-1,-.5) circle [radius=.05];
  \draw [fill=lightgray] (1,-.5) circle [radius=.05];

  \draw [fill=lightgray] (-1.5,-1) circle [radius=.05];
  \draw [fill=lightgray] (-.5,-1) circle [radius=.05];
  \draw [fill=lightgray] (.5,-1) circle [radius=.05];
  \draw [fill=lightgray] (1.5,-1) circle [radius=.05];

  \draw (-1.75,-1.5) circle [radius=.05];
  \draw (-1.25,-1.5) circle [radius=.05];
  \draw (-.75,-1.5) circle [radius=.05];
  \draw [fill=black] (-.25,-1.5) circle [radius=.05];
  \draw (.25,-1.5) circle [radius=.05];
  \draw [fill=black] (.75,-1.5) circle [radius=.05];
  \draw [fill=black] (1.25,-1.5) circle [radius=.05];
  \draw [fill=black] (1.75,-1.5) circle [radius=.05];

  \node at (-1.75,-1.7) {$\hat{u}_0$};
  \node at (-1.25,-1.7) {$\hat{u}_1$};
  \node at (-.75,-1.7) {$\hat{u}_2$};
  \node at (-.25,-1.7) {$\hat{u}_3$};
  \node at (.25,-1.7) {$\hat{u}_4$};
  \node at (.75,-1.7) {$\hat{u}_5$};
  \node at (1.25,-1.7) {$\hat{u}_6$};
  \node at (1.75,-1.7) {$\hat{u}_7$};

  \draw (0,-.05) -- (-1,-.45) node[pos=0.5,sloped,above] {$2$};
  \draw (0,-.05) -- (1,-.45) node[pos=0.5,sloped,above] {$2$};

  \draw (-1,-.55) -- (-1.5,-.95) node[pos=0.5,sloped,above] {$1$};
  \draw (-1,-.55) -- (-.5,-.95) node[pos=0.5,sloped,above] {$1$};
  \draw (1,-.55) -- (.5,-.95) node[pos=0.5,sloped,above] {$1$};
  \draw (1,-.55) -- (1.5,-.95) node[pos=0.5,sloped,above] {$1$};

  \draw (-1.5,-1.05) -- (-1.75,-1.45) node[pos=0.5,sloped,above] {$1$};
  \draw (-1.5,-1.05) -- (-1.25,-1.45) node[pos=0.5,sloped,above] {$1$};
  \draw (-.5,-1.05) -- (-.75,-1.45) node[pos=0.5,sloped,above] {$1$};
  \draw (-.5,-1.05) -- (-.25,-1.45) node[pos=0.5,sloped,above] {$1$};
  \draw (.5,-1.05) -- (.25,-1.45) node[pos=0.5,sloped,above] {$1$};
  \draw (.5,-1.05) -- (.75,-1.45) node[pos=0.5,sloped,above] {$1$};
  \draw (1.5,-1.05) -- (1.25,-1.45) node[pos=0.5,sloped,above] {$1$};
  \draw (1.5,-1.05) -- (1.75,-1.45) node[pos=0.5,sloped,above] {$1$};

  \draw [thin,dotted] (-2,0) -- (2,0);
  \draw [thin,dotted] (-2,-.5) -- (2,-.5);
  \draw [thin,dotted] (-2,-1) -- (2,-1);
  \draw [thin,dotted] (-2,-1.5) -- (2,-1.5);

\end{tikzpicture}
\caption{Semi-parallel ($P=2$).}
\label{fig:sc-decSP}
\end{subfigure}
\begin{subfigure}{.32\textwidth}
\centering
\begin{tikzpicture}[scale=1.2, thick]
  \draw [fill=lightgray] (0,0) circle [radius=.05];

  \draw [fill=lightgray] (-1,-.5) circle [radius=.05];
  \draw [fill=lightgray] (1,-.5) circle [radius=.05];

  \draw [fill=lightgray] (-1.5,-1) circle [radius=.05];
  \draw [fill=lightgray] (-.5,-1) circle [radius=.05];
  \draw [fill=lightgray] (.5,-1) circle [radius=.05];
  \draw [fill=lightgray] (1.5,-1) circle [radius=.05];

  \draw (-1.75,-1.5) circle [radius=.05];
  \draw (-1.25,-1.5) circle [radius=.05];
  \draw (-.75,-1.5) circle [radius=.05];
  \draw [fill=black] (-.25,-1.5) circle [radius=.05];
  \draw (.25,-1.5) circle [radius=.05];
  \draw [fill=black] (.75,-1.5) circle [radius=.05];
  \draw [fill=black] (1.25,-1.5) circle [radius=.05];
  \draw [fill=black] (1.75,-1.5) circle [radius=.05];

  \node at (-1.75,-1.7) {$\hat{u}_0$};
  \node at (-1.25,-1.7) {$\hat{u}_1$};
  \node at (-.75,-1.7) {$\hat{u}_2$};
  \node at (-.25,-1.7) {$\hat{u}_3$};
  \node at (.25,-1.7) {$\hat{u}_4$};
  \node at (.75,-1.7) {$\hat{u}_5$};
  \node at (1.25,-1.7) {$\hat{u}_6$};
  \node at (1.75,-1.7) {$\hat{u}_7$};

  \draw (0,-.05) -- (-1,-.45) node[pos=0.5,sloped,above] {$4$};
  \draw (0,-.05) -- (1,-.45) node[pos=0.5,sloped,above] {$4$};

  \draw (-1,-.55) -- (-1.5,-.95) node[pos=0.5,sloped,above] {$2$};
  \draw (-1,-.55) -- (-.5,-.95) node[pos=0.5,sloped,above] {$2$};
  \draw (1,-.55) -- (.5,-.95) node[pos=0.5,sloped,above] {$2$};
  \draw (1,-.55) -- (1.5,-.95) node[pos=0.5,sloped,above] {$2$};

  \draw (-1.5,-1.05) -- (-1.75,-1.45) node[pos=0.5,sloped,above] {$1$};
  \draw (-1.5,-1.05) -- (-1.25,-1.45) node[pos=0.5,sloped,above] {$1$};
  \draw (-.5,-1.05) -- (-.75,-1.45) node[pos=0.5,sloped,above] {$1$};
  \draw (-.5,-1.05) -- (-.25,-1.45) node[pos=0.5,sloped,above] {$1$};
  \draw (.5,-1.05) -- (.25,-1.45) node[pos=0.5,sloped,above] {$1$};
  \draw (.5,-1.05) -- (.75,-1.45) node[pos=0.5,sloped,above] {$1$};
  \draw (1.5,-1.05) -- (1.25,-1.45) node[pos=0.5,sloped,above] {$1$};
  \draw (1.5,-1.05) -- (1.75,-1.45) node[pos=0.5,sloped,above] {$1$};

  \draw [thin,dotted] (-2,0) -- (2,0);
  \draw [thin,dotted] (-2,-.5) -- (2,-.5);
  \draw [thin,dotted] (-2,-1) -- (2,-1);
  \draw [thin,dotted] (-2,-1.5) -- (2,-1.5);

\end{tikzpicture}
\caption{Fully-serial ($P=1$).}
\label{fig:sc-decFS}
\end{subfigure}
\caption{Decoding weights on a SC decoding tree for a polar code with $N=8$ and $R=1/2$.}
\label{fig:scDecLatency}
\end{figure}
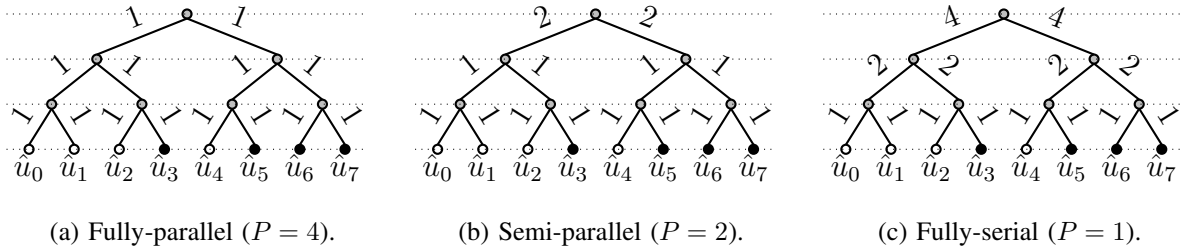

%- - - - - - - - - - - - - - - - - - - - - - - - - - - -
\subsection{Simplified Successive-Cancellation Decoding} \label{subsec:simplified}
%- - - - - - - - - - - - - - - - - - - - - - - - - - - -

The SSC decoding algorithm \cite{alamdar} identifies two types of nodes in the SC decoding tree. The bits within each node can be decoded efficiently in one shot without traversing its descendent nodes. These two types of nodes are:
\begin{itemize}
    \item \emph{Rate-0 node}: A node at level $s$ of the SC decoding tree all of whose leaf nodes at level $0$ are frozen bits. For a Rate-0 node at level $s$, bit estimates can be directly calculated at the level where the node is located as
    \begin{equation}
        \beta^i_s = 0 \text{.}
    \end{equation}
    \item \emph{Rate-1 node}: A node at level $s$ of the SC decoding tree whose leaf nodes at level $0$ are all information bits. For a Rate-1 node at level $s$, the bit estimations can be directly calculated at the level where the node is located as
    \begin{equation}
        \beta^i_s =
        \begin{cases}
        0 & \text{if $\alpha^i_s>0$,}\\
        1 & \text{otherwise.}
        \end{cases}
    \end{equation}
\end{itemize}

This paper considers a non-systematic polar code, whose information bits appear at level $0$. A non-systematic polar code requires hard decisions to calculate the information bits at level $0$ from the estimated bits at an intermediate level where a Rate-0 or a Rate-1 node is located. However, the bit-wise calculations are usually conducted in the same time step in which the LLR values are calculated, because the bit-wise calculations are much faster than the LLR calculations. Moreover, if a systematic polar code \cite{AriSys2011} (whose information bits appear at level $n$) is considered, there is no need to calculate the bit values at the leaf nodes because the information is present in the root node of the decoding tree. In fact, SSC decoding can decode Rate-0 and Rate-1 nodes in a single time step. In a binary tree representation of SC decoding, this corresponds to pruning all the nodes that are the descendants of a Rate-0 node or a Rate-1 node, as illustrated in Fig.~\ref{fig:sscDec}.

For practical code lengths, SSC decoding has a significantly lower latency than SC decoding \cite{alamdar}. This is due to the fact that the number of edges in the SSC decoding tree is significantly smaller than the number of edges in the SC decoding tree. Further, the latency of SSC decoding can be calculated by adding all the decoding weights in its (pruned) binary tree representation (as done in the case of SC decoding).

\begin{figure}[t]
\centering
\begin{tikzpicture}[scale=1.9, thick]
  \draw [fill=lightgray] (0,0) circle [radius=.05];

  \draw [fill=lightgray] (-1,-.5) circle [radius=.05];
  \draw [fill=lightgray] (1,-.5) circle [radius=.05];

  \draw (-1.5,-1) circle [radius=.05];
  \draw [fill=lightgray] (-.5,-1) circle [radius=.05];
  \draw [fill=lightgray] (.5,-1) circle [radius=.05];
  \draw [fill=black] (1.5,-1) circle [radius=.05];

  \draw (-.75,-1.5) circle [radius=.05];
  \draw [fill=black] (-.25,-1.5) circle [radius=.05];
  \draw (.25,-1.5) circle [radius=.05];
  \draw [fill=black] (.75,-1.5) circle [radius=.05];

  \draw (0,-.05) -- (-1,-.45);
  \draw (0,-.05) -- (1,-.45);

  \draw (-1,-.55) -- (-1.5,-.95);
  \draw (-1,-.55) -- (-.5,-.95);
  \draw (1,-.55) -- (.5,-.95);
  \draw (1,-.55) -- (1.5,-.95);

  \draw (-.5,-1.05) -- (-.75,-1.45);
  \draw (-.5,-1.05) -- (-.25,-1.45);
  \draw (.5,-1.05) -- (.25,-1.45);
  \draw (.5,-1.05) -- (.75,-1.45);

  \draw [thin,dotted] (-2,0) -- (2,0);
  \draw [thin,dotted] (-2,-.5) -- (2,-.5);
  \draw [thin,dotted] (-2,-1) -- (2,-1);
  \draw [thin,dotted] (-2,-1.5) -- (2,-1.5);

  \node at (2.25,0) {$s=3$};
  \node at (2.25,-.5) {$s=2$};
  \node at (2.25,-1) {$s=1$};
  \node at (2.25,-1.5) {$s=0$};
\end{tikzpicture}
\caption{Binary tree representation of SSC decoding for a polar code with $N=8$ and $R=1/2$. The white nodes represent Rate-0 nodes, the black nodes represent Rate-1 nodes, and the gray nodes are neither Rate-0 nodes nor Rate-1 nodes.}
\label{fig:sscDec}
\end{figure}
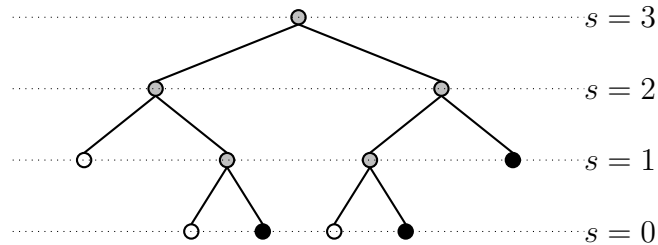

%----------------------------------------------------
\section{Latency of SSC Decoding with Limited Parallelism} \label{sec:main}
%----------------------------------------------------

%In this section, we present and prove our main results. First, Theorem~\ref{thm:decoding_complexity_ssc} gives an asymptotic upper bound for the decoding complexity of the SSC algorithm (or, equivalently, for the latency of a fully-serial implementation of the SSC decoder). Then, Theorem~\ref{thm:latency_p_parallel} proves a general bound on the latency of SSC decoding with access to $P$ parallel processing units. It is clear that one can deduce the decoding complexity by simply setting $P=1$ in the general setup. 

\begin{theorem}[Latency of SSC Decoder with Limited Parallelism]\label{thm:latency_p_parallel}
	Let $W$ be a given BMS channel with symmetric capacity $I(W)$. Fix $p_e$ and design a sequence of polar codes $\mathcal{C}_{\text{polar}}(p_e,W,N)$ of increasing block lengths with rates approaching $I(W)$, as per Definition \ref{def:construction}. Then, for any $\epsilon>0$, there exists $\bar{N}(\epsilon)$ such that, for any $N\ge \bar{N}(\epsilon)$, the latency of the SSC decoder with $P$ processing elements is upper bounded by
	\begin{align}\label{eq:mainres}
	c\,N^{1-1/\mu} + (2+\epsilon)\frac{N}{P}\log_2 \log_2 \frac{N}{P},
	\end{align}
	where $c>0$ is an absolute constant (independent of $N, P, p_e, \epsilon$ and $W$).
\end{theorem}

Some remarks are in order. First, note that, in a fully-serial implementation with $P=1$, the upper bound \eqref{eq:mainres} reduces to
\begin{align}
	(2+o(1))N\log_2\log_2 N.
	\label{eq:constant}
	\end{align}
Furthermore, if $P=N^{1/\mu}$, then \eqref{eq:mainres} is
	\begin{align}
	\tilde{O}(N^{1-1/\mu}),
\end{align}
where the $\tilde{O}$ notation hides (log-)logarithmic factors. Recall that the latency of a fully-parallel implementation of the SSC decoder is $O(N^{1-1/\mu})$, see Theorem~1 of \cite{mondelli2020sublinear}. Thus, another immediate consequence of Theorem~\ref{thm:latency_p_parallel} is that $P\sim N^{1/\mu}$ suffices to get roughly the same latency as $P=N/2$, and this is the smallest such $P$. 
	
The key idea of the proof is to look at various levels of the decoding tree and approximate the number of nodes whose corresponding bit-channels are already polarized beyond a certain threshold. Such nodes will be pruned, thus reducing the total weight of the tree. A similar idea (though with a different pruning strategy) appears in \cite{mondelli2020sublinear}. However, our earlier work in \cite{mondelli2020sublinear} considers only the fully-parallel setting where $P=N/2$.

Before proceeding with the proof, two intermediate lemmas are required. The first one is a two-sided version of the bound on $Z_n$, as defined in (\ref{eq:eqBMSC}), leading to Theorem~3 in \cite{MHU15unif-ieeeit}. Its proof appears in Appendix~\ref{app:prf}. 

\begin{lemma}[Refined bound on number of un-polarized channels]\label{lemma:unpolarized}
Let $W$ be a BMS channel and let 
$Z_n=Z(W_n)$ be the random process that tracks the Bhattacharyya parameter of $W_n$. Let $\mu$ be an upper bound on the scaling exponent according to Definition~\ref{def:upscal}. Fix $\gamma\in\left(\frac{1}{1+\mu}, 1\right)$. Then, for $n\ge 1$, 
\begin{equation}\label{eq:nummid}
    \mathbb P\left(Z_n\in \left[2^{-2^{n\gamma h_2^{(-1)}\left(\frac{\gamma(\mu+1)-1}{\gamma\mu}\right)}}, 1-2^{-2^{n\gamma h_2^{(-1)}\left(\frac{\gamma(\mu+1)-1}{\gamma\mu}\right)}}\right]\right) \le c\,2^{-n(1-\gamma)/\mu},
\end{equation}
where $c$ is a numerical constant that does not depend on $n$, $W$, or $\gamma$, and $h_2^{(-1)}$ is the inverse of the binary entropy function $h_2(x)=-x\log_2 x-(1-x)\log_2 (1-x)$ for $x\in [0, 1/2]$.
\end{lemma}

The second intermediate result is stated as Lemma~2 in \cite{mondelli2020sublinear}.

\begin{lemma}[Sufficient condition for Rate-0 and Rate-1 nodes]\label{lemma:rate01}
Let $W$ be a BMS channel, $p_{e}\in (0, 1)$, $N=2^n$, and $M=2^m$ with $m<n$. Consider the polar code $\mathcal C_{\rm polar}(p_{e}/M, W, N/M)$ constructed according to Definition~\ref{def:construction}. Then, there exists an integer $n_0$, which depends on $p_{e}$, such that for all $n\ge n_0$, the following holds:
\begin{enumerate}
    \item If $Z(W)\le 1/N^3$, then the polar code $\mathcal C_{\rm polar}(p_{e}/M, W, N/M)$ has rate $1$.
    \item If $Z(W)\ge 1-1/N^3$, then the polar code $\mathcal C_{\rm polar}(p_{e}/M, W, N/M)$ has rate $0$.
\end{enumerate} 
\end{lemma}

At this point, the proof of Theorem~\ref{thm:latency_p_parallel} is presented.

\begin{proof}[Proof of Theorem~\ref{thm:latency_p_parallel}]
	The decoding tree is divided into two segments. The first part is called $\mathcal{F}_1$ and it consists of all nodes/edges at distance at most $\lceil\log_2(N/P)\rceil$ from the root node.
	The second part is called $\mathcal{F}_2$ and it consists of the rest, which are all the nodes/edges in the bottom $\lfloor\log_2 P\rfloor$ layers. To analyze the latency, three cases are considered: ({\bf Case~A}) $N^{0.01}\leq P \leq N^{0.99}$ (moderate values of $P$), ({\bf Case~B}) $N^{0.99}\le P$ (large values of $P$), and ({\bf Case~C}) $P\le N^{0.01}$ (small values of $P$).

	%We denote this part by $\mathcal{F}_2$. Given that all of the edges in $\mathcal{F}_1$ have decoding weights larger than or equal to $P$ and the fact that we have access to $P$ parallel processing elements, we can divide these weights by $P$ in order to compute the corresponding delay for edge. This transforms $\mathcal{F}_1$ into the decoding tree of a polar code of length $N/P$ with the same weights as in the proof of Theorem~\ref{thm:decoding_complexity_ssc}. On the other hand, all of the edges in $\mathcal{F}_2$ have decoding weights less than or equal to $P$. Hence, by having access to $P$ parallel processing elements, the decoder can make all the computations in each such edge in one instance of time. Thus, these decoding weights are all equal to $1$, and the part $\mathcal{F}_2$ is equivalent to a collection of $N/P$ binary trees with decoding weights equal to $1$. This setup is similar to that of the latency analysis for the fully-parallel implementation in~\cite{mondelli2020sublinear}.
	%In fact, if there is no pruning (as in the simplified SC decoder), the decoding latency can be simply computed as
	%\begin{align}
	%\underbrace{\frac{N}{P}\log_2 \frac{N}{P}}_{\text{for }\mathcal{F}_1} + 
	%\underbrace{\frac{N}{P}(2P-2)}_{\text{for }\mathcal{F}_2}.
	%\end{align}
	
	\vspace{1em}
	
	\noindent {\bf Case~A:} $N^{0.01}\leq P \leq N^{0.99}$.
 Let us first look at $\mathcal{F}_1$, and consider pruning at depths $k_1$ and $k_1+k_2$, with
\begin{align}\label{eq:k1k2new}
\begin{split}
k_1 &= \bigg\lceil c_1 \log_2 \log_2 \frac{N}{P} \bigg\rceil,\\
k_2 &= \bigg\lceil c_2 \log_2 \log_2 \frac{N}{P} \bigg\rceil,
\end{split}
\end{align}
where $c_1$ and $c_2$ are constants to be determined later. Further assume that
\begin{align}
c_1\gamma_1h_2^{(-1)}\left(\frac{\gamma_1(\mu+1)-1}{\gamma_1\mu}\right) &> 1,\label{eq:suff_conditions1}\\
c_2\gamma_2h_2^{(-1)}\left(\frac{\gamma_2(\mu+1)-1}{\gamma_2\mu}\right) &> 1,\label{eq:suff_conditions2}
\end{align}
where the constants $\gamma_1$ and $\gamma_2$ will be also determined later.
If \eqref{eq:suff_conditions1} and \eqref{eq:suff_conditions2} are true,
then, as $P\leq N^{0.99}$, for sufficiently large values of $N$,
\begin{align}\label{eq:pruning_conditions_k1k2}
\begin{split}
2^{-2^{k_1\gamma_1 h_2^{(-1)}\left(\frac{\gamma_1(\mu+1)-1}{\gamma_1\mu}\right)}} &\leq \frac{1}{N^3},\\
2^{-2^{k_2\gamma_2 h_2^{(-1)}\left(\frac{\gamma_2(\mu+1)-1}{\gamma_2\mu}\right)}} &\leq \frac{1}{N^3}.
\end{split}
\end{align}
Also,
    \begin{align}\label{eq:limit}
        \lim_{\gamma_1\rightarrow 1} \gamma_1 h_2^{(-1)}\left(\frac{\gamma_1(\mu+1)-1}{\gamma_1\mu}\right) = \frac{1}{2}.
    \end{align}
We choose $c_1=2+\epsilon$ for a positive $\epsilon$. In view of (\ref{eq:limit}), there exists $\delta>0$ such that \eqref{eq:suff_conditions1} is satisfied by taking $\gamma_1 = 1-\delta$. Furthermore, we pick $\gamma_2 = 0.9$ and $c_2 =100$. Selecting $\mu\ge 2$ ensures that \eqref{eq:suff_conditions2} holds.

% Thus, by Lemma \ref{lemma:rate01}, the pruning operation at depth $k_1$ leaves at most $c2^{k_1(1-\frac{1-\gamma_1}{\mu})}$ nodes. Hence, for part \emph{(b)}, there are at most $c2^{k_1(1-\frac{1-\gamma_1}{\mu})}$ smaller trees, each with depth $k_2$ and total weight of $k_22^{n-k_1}$. At layer $k_2$, each of these sub-trees has a total of $2^{k_2}$ nodes before pruning. Now, we use a similar argument as in part (a). We invoke Lemma~\ref{lemma:unpolarized} with $n = k_2$ and $\gamma = \gamma_2$, and combine the second inequality in~(\ref{eq:pruning_conditions_k1k2}) with Lemma \ref{lemma:rate01}. Thus, we deduce that there will remain at most $c2^{k_2(1-\frac{1-\gamma_2}{\mu})}$ nodes in each of those sub-trees in the overall depth-$(k_1+k_2)$ upon the pruning operation. Hence, for part \emph{(c)}, the total number of sub-trees reduces to $c2^{k_1(1-\frac{1-\gamma_1}{\mu})}\cdot c2^{k_2(1-\frac{1-\gamma_2}{\mu})}$, each contributing a total weight of $(n-k_1-k_2)2^{n-k_1-k_2}$.

Now, the latency associated to $\mathcal F_1$ can be computed. To do so, $\mathcal F_1$ is partitioned into three parts: \emph{(i)} nodes that appear above depth $k_1$, \emph{(ii)} what remains between depth $k_1$ and the next $k_2$ layers after pruning the tree at layer $k_1$, and \emph{(iii)} what remains of $\mathcal F_1$ after pruning at depth $k_1+k_2$.

For part \emph{(i)}, the total decoding weight sums up to
\begin{equation}\label{eq:term1}
    \sum_{i=1}^{k_1} 2^i \bigg\lceil\frac{N}{2^i P}\bigg\rceil \le 2^{k_1+1}+ k_1\frac{N}{P}.
\end{equation}

At layer $k_1$, there are a total of $2^{k_1}$ nodes prior to the pruning. By using Lemma~\ref{lemma:unpolarized} and the first inequality in (\ref{eq:pruning_conditions_k1k2}), there are at most
\begin{equation}\label{eq:defa1}
a_1 \triangleq c2^{k_1(1-\frac{1-\gamma_1}{\mu})} \leq c2^{k_1}
\end{equation}
nodes whose Bhattacharyya parameter is in the interval $[1/N^3, 1-1/N^3]$. Thus, by applying Lemma~\ref{lemma:rate01} with $M=2^{k_1}$ and desired error probability set to $p_e \frac{2^{k_1}}{N}$, all but those $a_1$ nodes can be pruned. Hence, part \emph{(ii)} of $\mathcal F_1$ consists of at most $a_1$ sub-trees with depth $k_2$. Consequently, the total decoding weight for part \emph{(ii)} can be upper bounded by
\begin{equation}\label{eq:term2}
 a_1\, \sum_{i=1}^{k_2} 2^i \bigg\lceil\frac{N}{2^{i+k_1} P}\bigg\rceil \le  a_1\,2^{k_2+1}+ a_1\, k_2\frac{N}{P2^{k_1}}.   
\end{equation}

At layer $k_2$, each of the sub-trees has a total of $2^{k_2}$ nodes before pruning. By using Lemma~\ref{lemma:unpolarized} and the second inequality in (\ref{eq:pruning_conditions_k1k2}), at most $c2^{k_2(1-\frac{1-\gamma_2}{\mu})}$ of these nodes have Bhattacharyya parameter in the interval $[1/N^3, 1-1/N^3]$. Let $v$ denote one of these at most $c2^{k_2(1-\frac{1-\gamma_2}{\mu})}$ nodes, and consider the subtree rooted at $v$. If we descend $k_1$ layers in this subtree, there are a total of $2^{k_1}$ nodes in it prior to pruning. However, by Lemma~\ref{lemma:unpolarized} and the first inequality in (\ref{eq:pruning_conditions_k1k2}), at most $c^{k_1(1-\frac{1-\gamma_1}{\mu})}$ of these $2^{k_1}$ nodes have Bhattacharyya parameter in the interval $[1/N^3, 1-1/N^3]$. Thus, by applying Lemma~\ref{lemma:rate01} with $M=2^{k_1+k_2}$ and error probability set to $p_e \frac{2^{k_1+k_2}}{N}$, the number of remaining nodes after pruning at depth $k_1+k_2$ can be upper bounded by
	\begin{equation}\label{eq:defa2}
	    a_2 \triangleq c^2 2^{k_1(1-\frac{1-\gamma_1}{\mu})}2^{k_2(1-\frac{1-\gamma_2}{\mu})}.
	\end{equation}
	Consequently, the total decoding weight for part \emph{(iii)} can be upper bounded by 
	\begin{equation}\label{eq:term3}
\begin{split}
    	a_2\, \sum_{i=1}^{\lceil\log_2(N/P)\rceil-k_1-k_2} 2^i \bigg\lceil\frac{N}{2^{i+k_1+k_2} P}\bigg\rceil &\le  a_2\,2^{\lceil\log_2(N/P)\rceil-k_1-k_2+1}\\
    	&\hspace{1em}+
	 a_2\left(\bigg\lceil\log_2\left(\frac{N}{P}\right)\bigg\rceil - k_1 - k_2\right) \frac{N}{P2^{k_1+k_2}}  . 
	\end{split}	    
\end{equation}

As a result, the latency associated to $\mathcal F_1$ is upper bounded by the sum of the terms in \eqref{eq:term1}, \eqref{eq:term2}, and \eqref{eq:term3}. By using the definitions of $k_1$ and $k_2$ in \eqref{eq:k1k2new} and of $a_1$ and $a_2$ in \eqref{eq:defa1} and \eqref{eq:defa2}, after some algebraic manipulations,
	\begin{equation}\label{eq:expressions}
	    \begin{split}
	        &2^{k_1+1} \le  4\left(\log_2\frac{N}{P}\right)^{c_1},\\
	        &a_1\,2^{k_2+1}  \le c2^{k_1} 2^{k_2+1}\le 8c\left(\log_2\frac{N}{P}\right)^{c_1+c_2},\\
	        &a_1\, k_2\frac{N}{P2^{k_1}} =\frac{c\frac{N}{P}k_2}{2^{k_1\frac{1-\gamma_1}{\mu}}}\le \frac{c\,\frac{N}{P}\left(c_2\,\log_2 \log_2 \frac{N}{P}+1\right)}{\left(\log_2 \frac{N}{P}\right)^{c_1(1-\gamma_1)/\mu}},\\
	       &a_2\,2^{\lceil\log_2(N/P)\rceil-k_1-k_2+1} =\frac{2c^2\,2^{\lceil\log_2(N/P)\rceil}}{2^{k_1\frac{1-\gamma_1}{\mu}}2^{k_2\frac{1-\gamma_2}{\mu}}}\le\frac{4c^2\frac{N}{P}}{\left(\log_2 \frac{N}{P}\right)^{\frac{c_1(1-\gamma_1)}{\mu} + \frac{c_2(1-\gamma_2)}{\mu}}},\\
	    &a_2\left(\bigg\lceil\log_2\left(\frac{N}{P}\right)\bigg\rceil - k_1 - k_2\right) \frac{N}{P2^{k_1+k_2}} =  \frac{c^2\frac{N}{P}\left(\left\lceil\log_2\left(\frac{N}{P}\right)\right\rceil - k_1 - k_2\right)}{2^{k_1\frac{1-\gamma_1}{\mu}}2^{k_2\frac{1-\gamma_2}{\mu}}} \\
	    &\hspace{16.5em}\le  \frac{c^2\frac{N}{P}\log_2 \frac{N}{P}}{\left(\log_2 \frac{N}{P}\right)^{\frac{c_1(1-\gamma_1)}{\mu} + \frac{c_2(1-\gamma_2)}{\mu}}}.
	    \end{split}
	\end{equation}
Note that $\frac{c_1(1-\gamma_1)}{\mu} > 0$ and $\frac{c_2(1-\gamma_2)}{\mu} > 1$, while $\frac{N}{P}\geq N^{0.01}$. Thus, for large $N$, all the right hand sides of the expressions in \eqref{eq:expressions} are $o\left(\frac{N}{P}\log_2\log_2\frac{N}{P}\right)$, and the term $\frac{N}{P}k_1$ is the dominant one in the computation of the latency associated to $\mathcal F_1$. As a result, for sufficiently large $N$, this latency is upper bounded by 
	\begin{align}\label{eq:upper_bound_F_1_case_A}
	   (2+\epsilon)\frac{N}{P}\log_2 \log_2 \frac{N}{P},
	\end{align}
	for any $\epsilon>0$.
	
Let us now look at $\mathcal{F}_2$, where pruning starts at layer $k_3 = \lceil\log_2 \frac{N}{P}\rceil$. By applying Lemma~1 of \cite{mondelli2020sublinear} at level $k_3$, for any $\nu>1$, 
\begin{equation}\label{eq:battaoldbd}
    \mathbb P(Z_{k_3}\in [2^{-\nu k_3}, 1-2^{-\nu k_3}])\le c 2^{-k_3/\mu},
\end{equation}
where the constant $c$ depends solely on $\nu$ (and not on $k_3$ or $W$). Since $P\le N^{0.99}$, $k_3\ge 0.01\log_2 N$. Thus, by taking $\nu=300$ in \eqref{eq:battaoldbd}, at level $k_3$, the number of nodes whose Bhattacharyya parameter is in the interval $[1/N^3, 1-1/N^3]$ is at most 
\begin{equation}\label{eq:defa3}
a_3 \triangleq c_3\,2^{k_3(1-\frac{1}{\mu})},    
\end{equation}
for some constant $c_3$. Thus, by applying Lemma~\ref{lemma:rate01} with $M=2^{k_3}$ and error probability $\frac{p_e}{2^{n-k_3}}$, the number of remaining nodes after pruning at this layer can be upper bounded by $a_3$. Consequently, $\mathcal{F}_2$ consists of at most $a_3$ sub-trees of depth $\lfloor\log_2 P\rfloor$. Given that all nodes in $\mathcal{F}_2$ have decoding weights of $1$, the pruning strategy of \cite{mondelli2020sublinear} can be applied. Recall that $P\ge N^{0.01}$. Thus, by following the same strategy as in the proof of Theorem~1 in \cite{mondelli2020sublinear} and by boosting the constants $\nu$ by a factor of $100$, after pruning, each such sub-tree has a decoding weight of at most
	\begin{align}
	    c_4\, P^{1-\frac{1}{\mu}},
	\end{align}
for some constant $c_4$. Therefore, the decoding latency over $\mathcal{F}_2$ can be upper bounded by
	\begin{align}\label{eq:caseA_F2_upperbound}
	    a_3\, c_4\,P^{1-\frac{1}{\mu}} = c_5 N^{1-\frac{1}{\mu}},
	\end{align}
for some constant $c_5$. Combining the upper bounds in~(\ref{eq:upper_bound_F_1_case_A}) and~(\ref{eq:caseA_F2_upperbound}) concludes the proof for {\bf Case~A}.
	
		\vspace{1em}
	
	\noindent {\bf Case~B:} $N^{0.99} \leq P$. There is no need to prune part $\mathcal F_1$ of the tree. In fact, without any pruning, its latency is upper bounded by
	\begin{equation}\label{eq:smallpart}
	    \frac{N}{P}\log_2 \frac{N}{P} \le 0.01\,N^{0.01}\log_2 N.
	\end{equation}
	Part $\mathcal{F}_2$ starts at layer $k =  \lceil\log_2 \frac{N}{P}\rceil \leq \lceil 0.01 \cdot \log_2 N\rceil$. Recall that the decoding weights over $\mathcal{F}_2$ are all equal to $1$. Hence, the latency associated to $\mathcal{F}_2$ can be upper bounded by the decoding latency of the complete tree in a fully-parallel setup. This, in turn, is upper bounded by $cN^{1-\frac{1}{\mu}}$ for some universal constant $c>0$, see Theorem~1 of \cite{mondelli2020sublinear}. To conclude, note that the right hand side of \eqref{eq:smallpart} is smaller than $N^{1-\frac{1}{\mu}}$ for all sufficiently large $N$. Thus, the result for {\bf Case~B} readily follows.
	
		\vspace{1em}
	
	\noindent {\bf Case~C:} $P \leq N^{0.01}$. In this case, most of the latency is associated to $\mathcal{F}_1$. Recall that, when deriving the upper bound of the latency associated to $\mathcal{F}_1$ in {\bf Case~A}, the fact that $P \leq N^{0.99}$ is used, which is also satisfied in this case. Hence, by following the same argument as in {\bf Case~A}, for all sufficiently large $N$, the latency associated to $\mathcal{F}_1$ is upper bounded by
	\begin{align}\label{eq:upper_bound_F_1_case_C}
        (2+\epsilon)\frac{N}{P}\log_2 \log_2 \frac{N}{P},
    \end{align}
    for any $\epsilon>0$. Let us now look at $\mathcal{F}_2$. The tree is pruned at layer $k = \lceil\log_2 \frac{N}{P}\rceil\ge 0.99\log_2 N$. Thus, by applying \eqref{eq:battaoldbd} with $\nu=4$, at level $k$, the number of nodes whose Bhattacharyya parameter is in the interval $[1/N^3, 1-1/N^3]$ is at most $a\triangleq c(N/P)^{1-1/\mu}$, for some constant $c$. 
	Hence, by applying Lemma~\ref{lemma:rate01} with $M=2^{k}$, the number of remaining nodes after pruning at this layer can be upper bounded by $a$. Consequently,  $\mathcal{F}_2$ consists of at most $a$ many sub-trees of depth $\lfloor\log_2 P\rfloor$. Therefore, the latency associated to $\mathcal F_2$ is upper bounded by
	\begin{align}\label{eq:upper_bound_F_2_case_C}
       2 a P = o\left(\frac{N}{P}\right),
    \end{align}	
	where in the last step $P \leq N^{0.01}$ and $\mu \in [2,  5]$ are considered. This establishes the fact that~(\ref{eq:upper_bound_F_1_case_C}) is the dominant term in the computation of latency, which in turn completes the proof for {\bf Case~C}. 
\end{proof}

%----------------------------------------------------
\section{Numerical Results}\label{sec:numerical}
%----------------------------------------------------

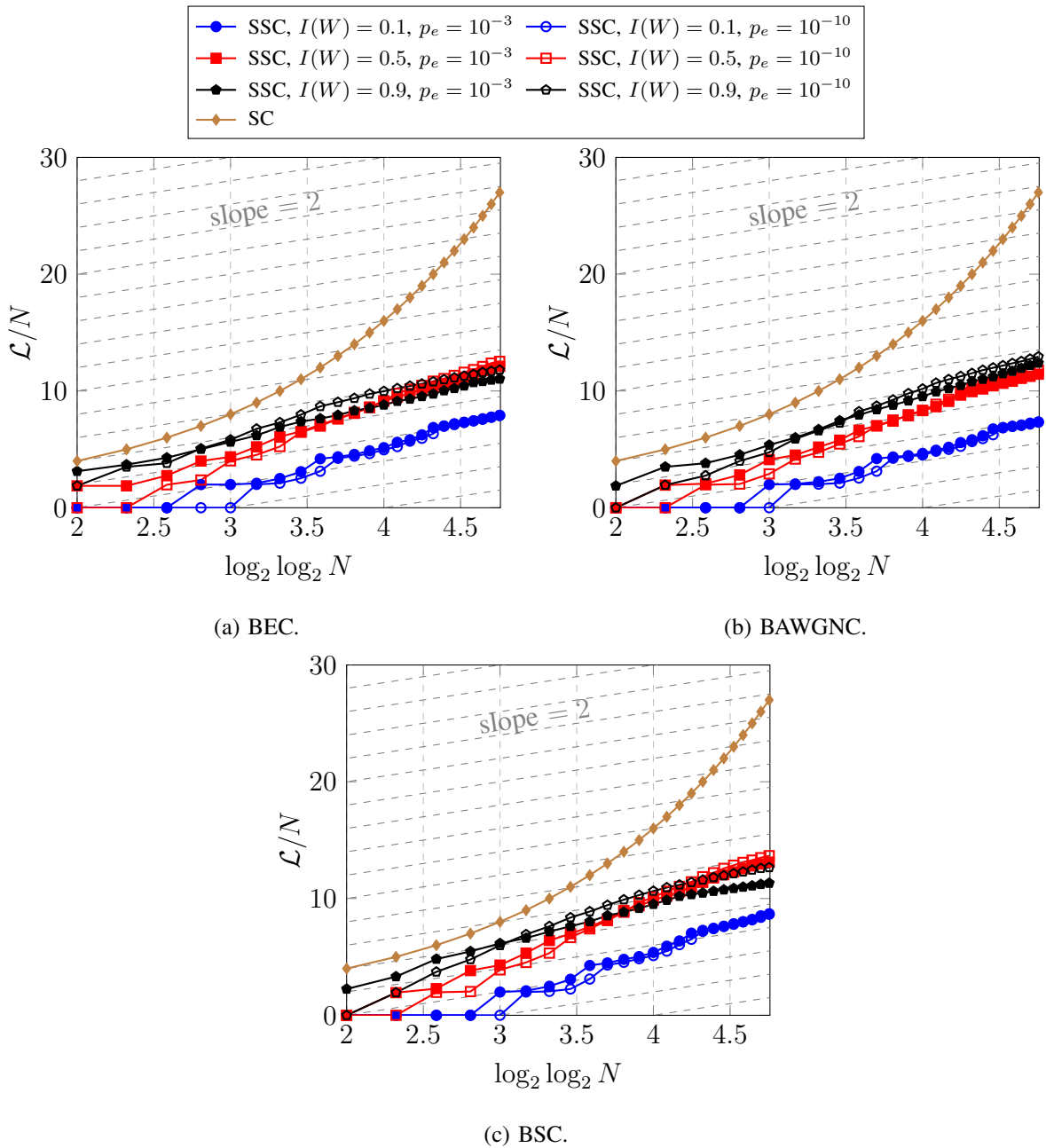
\begin{figure*}[t!]
\centering
\ref{legend-BECcomp}\\
\begin{subfigure}{0.48\textwidth}
\centering
\begin{tikzpicture}

\begin{axis}[
scale=1,
xmin=2,
xmax=4.76,
ymin=0,
ymax=30,
%ymajorgrids=true,
xmajorgrids=true,
grid style=dashed,
width=\textwidth, %height=7.5cm,
xlabel={$\log_2\log_2 N$},
ylabel={$\mathcal{L}/N$},
ylabel shift=-7,
legend cell align={left},
legend pos=north west,
legend style={
	column sep= 1mm,
	font=\fontsize{9pt}{9}\selectfont,
},
legend to name=legend-BECcomp,
legend columns=2,
]

\foreach \i in {-30,...,30}
{
\addplot[
color=gray,
dashed,
domain=-30:30,
samples=10,
smooth,
forget plot,
]
{2*x+2*\i};
}

\draw [color=gray, dashed, opacity=0] (1.5,20) -- (5,27) node [midway, above, sloped, gray, opacity=1] {slope $=2$};

\addplot[
color=blue,
mark=*,
thick
]
table {
0			0.000
1			0.000
1.584962501	0.000
2			0.000
2.321928095	0.000
2.584962501	0.000
2.807354922	1.984
3			1.984
3.169925001	2.094
3.321928095	2.504
3.459431619	3.088
3.584962501	4.192
3.700439718	4.361
3.807354922	4.563
3.906890596	4.877
4			5.144
4.087462841	5.606
4.169925001	5.828
4.247927513	6.242
4.321928095	6.870
4.392317423	7.004
4.459431619	7.191
4.523561956	7.310
4.584962501	7.444
4.64385619	7.619
4.700439718	7.758
4.754887502	7.913
};
\addlegendentry{SSC, $I(W)=0.1$, $p_e=10^{-3}$}

\addplot[
color=blue,
mark=o,
thick
]
table {
0			0.000
1			0.000
1.584962501	0.000
2			0.000
2.321928095	0.000
2.584962501	0.000
2.807354922	0.000
3			0.000
3.169925001	1.996
3.321928095	2.107
3.459431619	2.525
3.584962501	3.125
3.700439718	4.235
3.807354922	4.419
3.906890596	4.639
4			4.962
4.087462841	5.236
4.169925001	5.704
4.247927513	5.946
4.321928095	6.345
4.392317423	6.981
4.459431619	7.119
4.523561956	7.312
4.584962501	7.440
4.64385619	7.573
4.700439718	7.754
4.754887502	7.900
};
\addlegendentry{SSC, $I(W)=0.1$, $p_e=10^{-10}$}

\addplot[
color=red,
mark=square*,
thick
]
table {
0			 0.000
1			 0.000
1.584962501	 0.000
2			 1.875
2.321928095	 1.875
2.584962501	 2.750
2.807354922	 4.000
3			 4.359
3.169925001	 5.219
3.321928095	 6.084
3.459431619	 6.542
3.584962501	 7.099
3.700439718	 7.624
3.807354922	 8.084
3.906890596	 8.571
4			 9.143
4.087462841	 9.485
4.169925001	 9.788
4.247927513	10.136
4.321928095	10.364
4.392317423	10.616
4.459431619	10.921
4.523561956	11.180
4.584962501	11.459
4.64385619	11.751
4.700439718	11.929
4.754887502	12.109
};
\addlegendentry{SSC, $I(W)=0.5$, $p_e=10^{-3}$}

\addplot[
color=red,
mark=square,
thick
]
table {
0			 0.000
1			 0.000
1.584962501	 0.000
2			 0.000
2.321928095	 0.000
2.584962501	 1.969
2.807354922	 2.375
3			 4.016
3.169925001	 4.527
3.321928095	 5.240
3.459431619	 6.461
3.584962501	 6.994
3.700439718	 7.613
3.807354922	 8.162
3.906890596	 8.629
4			 9.154
4.087462841	 9.793
4.169925001	10.154
4.247927513	10.478
4.321928095	10.838
4.392317423	11.060
4.459431619	11.333
4.523561956	11.601
4.584962501	11.828
4.64385619	12.092
4.700439718	12.385
4.754887502	12.534
};
\addlegendentry{SSC, $I(W)=0.5$, $p_e=10^{-10}$}

\addplot[
color=black,
mark=pentagon*,
thick
]
table {
0			 0.000
1			 1.500
1.584962501	 1.750
2			 3.125
2.321928095	 3.688
2.584962501	 4.281
2.807354922	 5.016
3			 5.633
3.169925001	 6.188
3.321928095	 6.916
3.459431619	 7.378
3.584962501	 7.634
3.700439718	 7.929
3.807354922	 8.313
3.906890596	 8.531
4			 8.806
4.087462841	 9.121
4.169925001	 9.311
4.247927513	 9.534
4.321928095	 9.748
4.392317423	10.027
4.459431619	10.220
4.523561956	10.440
4.584962501	10.730
4.64385619	10.824
4.700439718	10.944
4.754887502	11.032
};
\addlegendentry{SSC, $I(W)=0.9$, $p_e=10^{-3}$}

\addplot[
color=black,
mark=pentagon,
thick
]
table {
0			 0.000
1			 0.000
1.584962501	 0.000
2			 1.875
2.321928095	 3.500
2.584962501	 3.812
2.807354922	 5.062
3			 5.820
3.169925001	 6.758
3.321928095	 7.303
3.459431619	 7.970
3.584962501	 8.677
3.700439718	 9.039
3.807354922	 9.377
3.906890596	 9.746
4			 9.977
4.087462841	10.220
4.169925001	10.425
4.247927513	10.612
4.321928095	10.793
4.392317423	10.980
4.459431619	11.140
4.523561956	11.281
4.584962501	11.423
4.64385619	11.591
4.700439718	11.695
4.754887502	11.798
};
\addlegendentry{SSC, $I(W)=0.9$, $p_e=10^{-10}$}

\addplot[
color=brown,
mark=diamond*,
thick
]
table {
0			1
1			2
1.584962501	3
2			4
2.321928095	5
2.584962501	6
2.807354922	7
3			8
3.169925001	9
3.321928095	10
3.459431619	11
3.584962501	12
3.700439718	13
3.807354922	14
3.906890596	15
4			16
4.087462841	17
4.169925001	18
4.247927513	19
4.321928095	20
4.392317423	21
4.459431619	22
4.523561956	23
4.584962501	24
4.64385619	25
4.700439718	26
4.754887502	27
};
\addlegendentry{SC}

\end{axis}
\end{tikzpicture}
\caption{BEC.}
\label{fig:NormLatency:bec}
\end{subfigure}
\begin{subfigure}{0.48\textwidth}
\centering
\begin{tikzpicture}

\begin{axis}[
scale=1,
xmin=2,
xmax=4.76,
ymin=0,
ymax=30,
%ymajorgrids=true,
xmajorgrids=true,
grid style=dashed,
width=\textwidth, %height=7.5cm,
xlabel={$\log_2\log_2 N$},
ylabel={$\mathcal{L}/N$},
ylabel shift=-7,
legend cell align={left},
legend pos=north west,
legend style={
	column sep= 1mm,
	font=\fontsize{9pt}{9}\selectfont,
},
%legend to name=legend-BECcomp,
%legend columns=2,
]

\foreach \i in {-30,...,30}
{
\addplot[
color=gray,
dashed,
domain=-30:30,
samples=10,
smooth,
forget plot,
]
{2*x+2*\i};
}

\draw [color=gray, dashed, opacity=0] (1.5,20) -- (5,27) node [midway, above, sloped, gray, opacity=1] {slope $=2$};

\addplot[
color=blue,
mark=*,
thick
]
table {
0			0.000
1			0.000
1.584962501	0.000
2			0.000
2.321928095	0.000
2.584962501	0.000
2.807354922	0.000
3			1.992
3.169925001	2.035
3.321928095	2.225
3.459431619	2.520
3.584962501	3.104
3.700439718	4.210
3.807354922	4.346
3.906890596	4.460
4			4.650
4.087462841	4.928
4.169925001	5.177
4.247927513	5.581
4.321928095	5.824
4.392317423	6.165
4.459431619	6.751
4.523561956	6.867
4.584962501	6.964
4.64385619	7.082
4.700439718	7.242
4.754887502	7.341
};
%\addlegendentry{SSC, $C=0.1$, $P_e=10^{-3}$}

\addplot[
color=blue,
mark=o,
thick
]
table {
0			0.000
1			0.000
1.584962501	0.000
2			0.000
2.321928095	0.000
2.584962501	0.000
2.807354922	0.000
3			0.000
3.169925001	1.996
3.321928095	2.018
3.459431619	2.115
3.584962501	2.543
3.700439718	3.132
3.807354922	4.243
3.906890596	4.391
4			4.518
4.087462841	4.841
4.169925001	5.002
4.247927513	5.258
4.321928095	5.669
4.392317423	5.921
4.459431619	6.253
4.523561956	6.849
4.584962501	6.972
4.64385619	7.052
4.700439718	7.171
4.754887502	7.334
};
%\addlegendentry{SSC, $C=0.1$, $P_e=10^{-10}$}

\addplot[
color=red,
mark=square*,
thick
]
table {
0			 0.000
1			 0.000
1.584962501	 0.000
2			 0.000
2.321928095	 1.938
2.584962501	 2.062
2.807354922	 2.812
3			 4.125
3.169925001	 4.508
3.321928095	 5.195
3.459431619	 5.785
3.584962501	 6.644
3.700439718	 7.010
3.807354922	 7.511
3.906890596	 7.879
4			 8.319
4.087462841	 8.651
4.169925001	 9.106
4.247927513	 9.644
4.321928095	 9.942
4.392317423	10.170
4.459431619	10.443
4.523561956	10.659
4.584962501	10.832
4.64385619	11.079
4.700439718	11.253
4.754887502	11.434
};
%\addlegendentry{SSC, $C=0.5$, $P_e=10^{-3}$}

\addplot[
color=red,
mark=square,
thick
]
table {
0			 0.000
1			 0.000
1.584962501	 0.000
2			 0.000
2.321928095	 0.000
2.584962501	 1.969
2.807354922	 2.031
3			 2.898
3.169925001	 4.172
3.321928095	 4.752
3.459431619	 5.461
3.584962501	 6.099
3.700439718	 7.021
3.807354922	 7.416
3.906890596	 7.959
4			 8.333
4.087462841	 8.859
4.169925001	 9.288
4.247927513	 9.857
4.321928095	10.148
4.392317423	10.446
4.459431619	10.679
4.523561956	10.904
4.584962501	11.117
4.64385619	11.363
4.700439718	11.568
4.754887502	11.772
};
%\addlegendentry{SSC, $C=0.5$, $P_e=10^{-10}$}

\addplot[
color=black,
mark=pentagon*,
thick
]
table {
0			 0.000
1			 0.000
1.584962501	 1.750
2			 1.875
2.321928095	 3.500
2.584962501	 3.812
2.807354922	 4.531
3			 5.375
3.169925001	 5.996
3.321928095	 6.701
3.459431619	 7.460
3.584962501	 7.946
3.700439718	 8.405
3.807354922	 8.793
3.906890596	 9.157
4			 9.528
4.087462841	 9.936
4.169925001	10.215
4.247927513	10.484
4.321928095	10.805
4.392317423	11.010
4.459431619	11.227
4.523561956	11.485
4.584962501	11.695
4.64385619	11.946
4.700439718	12.189
4.754887502	12.395
};
%\addlegendentry{SSC, $C=0.9$, $P_e=10^{-3}$}

\addplot[
color=black,
mark=pentagon,
thick
]
table {
0			 0.000
1			 0.000
1.584962501	 0.000
2			 0.000
2.321928095	 1.938
2.584962501	 2.750
2.807354922	 4.000
3			 4.742
3.169925001	 5.895
3.321928095	 6.600
3.459431619	 7.260
3.584962501	 8.233
3.700439718	 8.714
3.807354922	 9.246
3.906890596	 9.789
4			10.184
4.087462841	10.681
4.169925001	10.991
4.247927513	11.254
4.321928095	11.519
4.392317423	11.785
4.459431619	11.986
4.523561956	12.174
4.584962501	12.369
4.64385619	12.537
4.700439718	12.741
4.754887502	12.934
};
%\addlegendentry{SSC, $C=0.9$, $P_e=10^{-10}$}

\addplot[
color=brown,
mark=diamond*,
thick
]
table {
0			1
1			2
1.584962501	3
2			4
2.321928095	5
2.584962501	6
2.807354922	7
3			8
3.169925001	9
3.321928095	10
3.459431619	11
3.584962501	12
3.700439718	13
3.807354922	14
3.906890596	15
4			16
4.087462841	17
4.169925001	18
4.247927513	19
4.321928095	20
4.392317423	21
4.459431619	22
4.523561956	23
4.584962501	24
4.64385619	25
4.700439718	26
4.754887502	27
};
%\addlegendentry{SC}

\end{axis}
\end{tikzpicture}
\caption{BAWGNC.}
\label{fig:NormLatency:bawgnc}
\end{subfigure}
\begin{subfigure}{0.48\textwidth}
\centering
\begin{tikzpicture}

\begin{axis}[
scale=1,
xmin=2,
xmax=4.76,
ymin=0,
ymax=30,
%ymajorgrids=true,
xmajorgrids=true,
grid style=dashed,
width=\textwidth, %height=7.5cm,
xlabel={$\log_2\log_2 N$},
ylabel={$\mathcal{L}/N$},
ylabel shift=-7,
legend cell align={left},
legend pos=north west,
legend style={
	column sep= 1mm,
	font=\fontsize{9pt}{9}\selectfont,
},
%legend to name=legend-BECcomp,
%legend columns=2,
]

\foreach \i in {-30,...,30}
{
\addplot[
color=gray,
dashed,
domain=-30:30,
samples=10,
smooth,
forget plot,
]
{2*x+2*\i};
}

\draw [color=gray, dashed, opacity=0] (1.5,20) -- (5,27) node [midway, above, sloped, gray, opacity=1] {slope $=2$};

\addplot[
color=blue,
mark=*,
thick
]
table {
0			0.000
1			0.000
1.584962501	0.000
2			0.000
2.321928095	0.000
2.584962501	0.000
2.807354922	0.000
3			1.992
3.169925001	2.094
3.321928095	2.484
3.459431619	3.084
3.584962501	4.264
3.700439718	4.484
3.807354922	4.788
3.906890596	5.024
4			5.383
4.087462841	5.917
4.169925001	6.366
4.247927513	7.031
4.321928095	7.270
4.392317423	7.424
4.459431619	7.617
4.523561956	7.839
4.584962501	7.999
4.64385619	8.212
4.700439718	8.517
4.754887502	8.652
};
%\addlegendentry{SSC, $C=0.1$, $P_e=10^{-3}$}

\addplot[
color=blue,
mark=o,
thick
]
table {
0			0.000
1			0.000
1.584962501	0.000
2			0.000
2.321928095	0.000
2.584962501	0.000
2.807354922	0.000
3			0.000
3.169925001	1.996
3.321928095	2.047
3.459431619	2.249
3.584962501	3.103
3.700439718	4.307
3.807354922	4.537
3.906890596	4.846
4			5.115
4.087462841	5.497
4.169925001	6.049
4.247927513	6.508
4.321928095	7.192
4.392317423	7.446
4.459431619	7.614
4.523561956	7.793
4.584962501	8.013
4.64385619	8.175
4.700439718	8.391
4.754887502	8.692
};
%\addlegendentry{SSC, $C=0.1$, $P_e=10^{-10}$}

\addplot[
color=red,
mark=square*,
thick
]
table {
0			 0.000
1			 0.000
1.584962501	 0.000
2			 0.000
2.321928095	 1.938
2.584962501	 2.281
2.807354922	 3.828
3			 4.297
3.169925001	 5.328
3.321928095	 6.402
3.459431619	 6.985
3.584962501	 7.575
3.700439718	 8.217
3.807354922	 8.983
3.906890596	 9.381
4			 9.861
4.087462841	10.213
4.169925001	10.574
4.247927513	10.997
4.321928095	11.382
4.392317423	11.762
4.459431619	12.038
4.523561956	12.282
4.584962501	12.550
4.64385619	12.787
4.700439718	12.953
4.754887502	13.115
};
%\addlegendentry{SSC, $C=0.5$, $P_e=10^{-3}$}

\addplot[
color=red,
mark=square,
thick
]
table {
0			 0.000
1			 0.000
1.584962501	 0.000
2			 0.000
2.321928095	 0.000
2.584962501	 1.969
2.807354922	 2.031
3			 3.891
3.169925001	 4.527
3.321928095	 5.322
3.459431619	 6.660
3.584962501	 7.405
3.700439718	 8.117
3.807354922	 8.829
3.906890596	 9.655
4			10.099
4.087462841	10.624
4.169925001	11.034
4.247927513	11.427
4.321928095	11.865
4.392317423	12.203
4.459431619	12.578
4.523561956	12.850
4.584962501	13.091
4.64385619	13.286
4.700439718	13.493
4.754887502	13.667
};
%\addlegendentry{SSC, $C=0.5$, $P_e=10^{-10}$}

\addplot[
color=black,
mark=pentagon*,
thick
]
table {
0			 0.000
1			 0.000
1.584962501	 1.750
2			 2.250
2.321928095	 3.312
2.584962501	 4.812
2.807354922	 5.469
3			 6.156
3.169925001	 6.613
3.321928095	 7.166
3.459431619	 7.623
3.584962501	 8.030
3.700439718	 8.536
3.807354922	 8.833
3.906890596	 9.166
4			 9.510
4.087462841	 9.858
4.169925001	10.196
4.247927513	10.343
4.321928095	10.477
4.392317423	10.624
4.459431619	10.748
4.523561956	10.867
4.584962501	10.984
4.64385619	11.090
4.700439718	11.218
4.754887502	11.317
};
%\addlegendentry{SSC, $C=0.9$, $P_e=10^{-3}$}

\addplot[
color=black,
mark=pentagon,
thick
]
table {
0			 0.000
1			 0.000
1.584962501	 0.000
2			 0.000
2.321928095	 1.938
2.584962501	 3.719
2.807354922	 4.781
3			 5.984
3.169925001	 6.922
3.321928095	 7.619
3.459431619	 8.367
3.584962501	 8.883
3.700439718	 9.432
3.807354922	 9.901
3.906890596	10.293
4			10.638
4.087462841	10.940
4.169925001	11.173
4.247927513	11.381
4.321928095	11.594
4.392317423	11.801
4.459431619	11.960
4.523561956	12.136
4.584962501	12.285
4.64385619	12.449
4.700439718	12.590
4.754887502	12.659
};
%\addlegendentry{SSC, $C=0.9$, $P_e=10^{-10}$}

\addplot[
color=brown,
mark=diamond*,
thick
]
table {
0			1
1			2
1.584962501	3
2			4
2.321928095	5
2.584962501	6
2.807354922	7
3			8
3.169925001	9
3.321928095	10
3.459431619	11
3.584962501	12
3.700439718	13
3.807354922	14
3.906890596	15
4			16
4.087462841	17
4.169925001	18
4.247927513	19
4.321928095	20
4.392317423	21
4.459431619	22
4.523561956	23
4.584962501	24
4.64385619	25
4.700439718	26
4.754887502	27
};
%\addlegendentry{SC}

\end{axis}
\end{tikzpicture}
\caption{BSC.}
\label{fig:NormLatency:bsc}
\end{subfigure}
\caption{Normalized latency of SC and SSC decoding of polar codes in a fully-serial implementation ($P=1$). As the code length $N$ increases, the slope of the curves for SSC decoding tends to $2$, confirming that the latency of the simplified decoder scales as $(2+o(1))N\log_2\log_2 N$.}
\label{fig:NormLatency}
\end{figure*}

This section numerically evaluates SSC-decoding latency for polar codes, constructed based on Definition~\ref{def:construction} with $4\leq \log_2 N\leq 27$, when a limited number of PEs are available. To illustrate SSC-decoding latency in a fully-serial implementation ($P=1$), Fig.~\ref{fig:NormLatency} plots the latency normalized with respect to the block length $N$, namely $\mathcal{L}/N$ (on the $y$-axis) versus $\log_2\log_2 N$ (on the $x$-axis) when $I(W)\in\{0.1,0.5,0.9\}$ and $p_e\in\{10^{-3},10^{-10}\}$ for BEC (Fig.~\ref{fig:NormLatency:bec}), BAWGNC (Fig.~\ref{fig:NormLatency:bawgnc}), and BSC (Fig.~\ref{fig:NormLatency:bsc}). These figures show that SSC decoder's normalized decoding latency grows linearly with $\log_2\log_2 N$, confirming Theorem~\ref{thm:latency_p_parallel}'s upper bound (see (\ref{eq:constant})). Moreover, the curves' slope approaches $2$, as predicted by our theoretical result. The normalized latency of SC decoding grows exponentially in the $\log_2\log_2 N$ domain because the SC decoder has a latency of $N\log_2 N$ when $P=1$.

\begin{figure}[t]
\centering
\begin{tikzpicture}

\begin{axis}[
scale=1,
xmin=4,
xmax=27,
ymin=0,
ymax=31,
%ymajorgrids=true,
xmajorgrids=true,
grid style=dashed,
width=.48\textwidth,
xlabel={$\log_2 N$},
ylabel={$\log_2 \mathcal{L}$},
ylabel shift=-7,
legend cell align={left},
legend pos=north west,
set layers=standard,
legend style={
	column sep= 1mm,
	font=\fontsize{9pt}{9}\selectfont,
	draw=none,
	fill opacity=0.75,
	text opacity = 1,
},
clip mode=individual,
%legend to name=legend-BECcomp,
%legend columns=2,
]

\foreach \i in {-30,...,30}
{
\addplot[
color=gray,
dashed,
domain=-30:30,
samples=10,
smooth,
forget plot,
]
{.72*x+2.5*\i};
}

\draw [color=gray, dashed, opacity=0] (17,1.1) -- (27,8.3) node [midway, above, sloped, gray, opacity=1] {slope $=0.72$};

\addplot[
color=blue,
mark=*,
thick
]
table {
%0   0.000
%1   0.000
%2   0.000
%3   0.000
4   3.000
5   3.000
6   4.907
7   6.000
8   6.807
9   7.755
10  8.672
11  9.555
12 10.416
13 11.242
14 12.046
15 12.851
16 13.639
17 14.430
18 15.217
19 15.992
20 16.761
21 17.528
22 18.289
23 19.047
24 19.801
25 20.556
26 21.306
27 22.056
};
\addlegendentry{$P=\frac{N}{2}$}

\addplot[
color=red,
mark=square*,
thick
]
table {
%0   0.000
%1   0.000
%2   0.000
%3   0.000
4   3.322
5   4.000
6   5.248
7   6.459
8   7.170
9   8.140
10  8.977
11  9.861
12 10.653
13 11.479
14 12.236
15 13.038
16 13.796
17 14.578
18 15.337
19 16.105
20 16.854
21 17.614
22 18.361
23 19.113
24 19.857
25 20.607
26 21.349
27 22.094
};
\addlegendentry{$P=N^{\frac{1}{2}}$}

\addplot[
color=black,
mark=triangle*,
thick
]
table {
%0   0.000
%1   0.000
%2   0.000
%3   0.000
4   4.000
5   4.907
6   6.209
7   7.672
8   8.295
9   9.349
10 10.326
11 10.997
12 11.987
13 12.819
14 13.632
15 14.398
16 15.230
17 16.038
18 16.819
19 17.574
20 18.346
21 19.118
22 19.876
23 20.650
24 21.417
25 22.182
26 22.921
27 23.672
};
\addlegendentry{$P=N^{\frac{1}{\mu}}$}

%\addplot
%table {
%0   0.000
%1   0.000
%2   0.000
%3   0.000
%4   4.000
%5   4.907
%6   6.524
%7   7.672
%8   8.295
%9   9.512
%10 10.522
%11 11.392
%12 12.068
%13 13.037
%14 13.848
%15 14.709
%16 15.471
%17 16.309
%18 17.144
%19 17.963
%20 18.685
%21 19.491
%22 20.290
%23 21.091
%24 21.839
%25 22.635
%26 23.414
%27 24.189
%};
%\addlegendentry{$P=N^{\frac{1}{4}}$}

\addplot[
color=brown,
mark=pentagon*,
thick,
on layer={axis foreground}
]
table {
%0   0.000
%1   0.000
%2   0.000
%3   0.000
4   4.907
5   5.907
6   7.459
7   9.000
8   9.160
9  10.409
10 11.628
11 12.729
12 13.844
13 14.430
14 15.500
15 16.572
16 17.228
17 18.275
18 19.316
19 20.067
20 21.091
21 21.866
22 22.900
23 23.716
24 24.544
25 25.575
26 26.432
27 27.301
};
\addlegendentry{$P=N^{\frac{1}{8}}$}

\addplot[
color=cyan,
mark=diamond*,
thick,
on layer={axis foreground}
]
table {
%0   0.000
%1   0.000
%2   0.000
%3   0.000
4   4.907
5   5.907
6   7.459
7   9.000
8  10.124
9  11.384
10 12.605
11 13.710
12 14.828
13 15.931
14 17.015
15 18.099
16 19.193
17 20.246
18 21.291
19 22.341
20 23.374
21 24.408
22 25.449
23 26.483
24 27.518
25 28.555
26 29.576
27 30.598
};
\addlegendentry{$P=1$}

\end{axis}
\end{tikzpicture}
\caption{Latency of SSC decoding of a polar code constructed for a BEC with $I(W)=0.5$ and $p_e=10^{-3}$ considering different values of $P$. The slope of the curve when $P=N^{\frac{1}{\mu}}$ is $1-\frac{1}{\mu}=0.72$ and is similar to the case where $P=\frac{N}{2}$.}
\label{fig:becLatResult}
\end{figure}
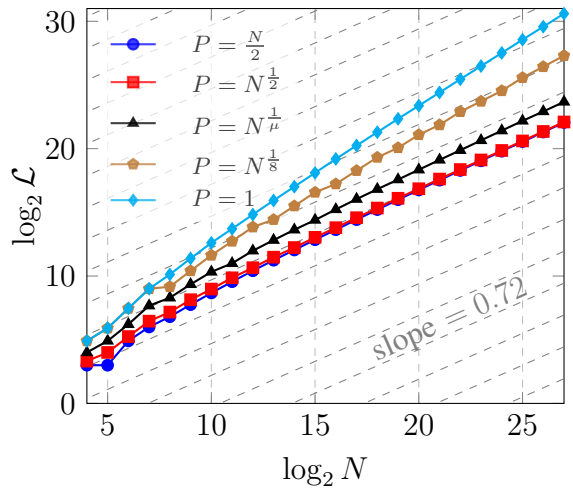

Fig.~\ref{fig:becLatResult} shows the SSC-decoding latency with $P \in \{1,N^{\frac{1}{8}},N^{\frac{1}{\mu}},N^{\frac{1}{2}},\frac{N}{2}\}$. The polar codes are constructed for a BEC with $I(W)=0.5$ and $p_e=10^{-3}$. It can be seen that, as $N$ increases, the slope of the curve with $P=N^{\frac{1}{\mu}}$ approaches $1-\frac{1}{\mu}$, which is $0.72$ for the BEC since $\mu \approx 3.63$ in this case. This scaling is the same as the lowest achievable latency when $P=\frac{N}{2}$. %As a result, $P=N^{\frac{1}{\mu}}$ is the optimal value to trade off latency with implementation complexity.

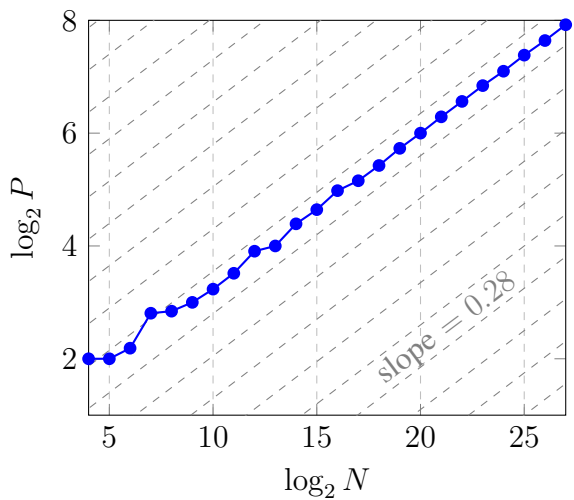
\begin{figure}[t]
\centering
\begin{tikzpicture}

\begin{axis}[
scale=1,
xmin=4,
xmax=27,
ymin=1,
ymax=8,
%ymajorgrids=true,
xmajorgrids=true,
grid style=dashed,
width=.48\textwidth,
xlabel={$\log_2 N$},
ylabel={$\log_2 P$},
%ylabel shift=-7,
legend cell align={left},
legend pos=north west,
legend style={
	column sep= 1mm,
	font=\fontsize{9pt}{9}\selectfont,
},
%legend to name=legend-BECcomp,
%legend columns=2,
]

\foreach \i in {-30,...,30}
{
\addplot[
color=gray,
dashed,
domain=-30:30,
samples=10,
smooth,
forget plot,
]
{.28*x+.75*\i};
}

\draw [color=gray, dashed, opacity=0] (17,.8) -- (27,3.55) node [midway, above, sloped, gray, opacity=1] {slope $=0.28$};

\addplot[
color=blue,
mark=*,
thick,
]
table {
0   0.000
1   0.000
2   0.000
3   0.000
4   2
5   1.9999980927
6   2.1875000000
7   2.8073530197
8   2.8437500000
9   2.9999961853
10  3.2343750000
11  3.5156250000
12  3.9068884850
13  4.0000019073
14  4.3923115730
15  4.6438512802
16  4.9804687500
17  5.1562500000
18  5.4262709618
19  5.7304687500
20  6.0000042915
21  6.2890625000
22  6.5625000000
23  6.8427734375
24  7.0976562500
25  7.3828125000
26  7.6416015625
27  7.9218750000
};
%\addlegendentry{$P=\frac{N}{2}$}

\end{axis}
\end{tikzpicture}
\caption{Required value of $P$ to achieve a latency for SSC decoding that is $1\%$ more than the fully-parallel implementation ($P=\frac{N}{2}$). Polar codes are constructed for a BEC with $I(W)=0.5$ and $p_e=10^{-3}$. The slope of the curve is $\frac{1}{\mu}=0.28$.}
\label{fig:becPResult}
\end{figure}

Fig.~\ref{fig:becPResult} shows how $P$ scales as $N$ increases when SSC-decoder latency is only $1\%$ higher than fully-parallel SSC decoding (i.e., the latency for $P=\frac{N}{2}$). The polar codes at different block lengths are constructed for a BEC with $I(W)=0.5$ and $p_e=10^{-3}$. Theorem \ref{thm:latency_p_parallel} predicts that, if $P$ scales as $N^{\frac{1}{\mu}}$, then the latency is close to that of the fully-parallel implementation, which Fig.~\ref{fig:becPResult} confirms because the curve's slope is $\frac{1}{\mu}=0.28$.

%----------------------------------------------------
\section{Summary}\label{sec:concl}
%----------------------------------------------------

This paper characterizes the latency of simplified successive-cancellation (SSC) decoding when there is a limited number of processing elements available to implement the decoder. We show that for a polar code of block length $N$, when the number of processing elements $P$ is limited, the latency of SSC decoding is $O(N^{1-1/\mu}+\frac{N}{P}\log_2\log_2\frac{N}{P})$, where $\mu$ is the scaling exponent of the channel. The bound resulted in three important implications. First, a fully-parallel implementation with $P=\frac{N}{2}$ results in a sublinear latency for SSC decoding, which recovers the result in \cite{mondelli2020sublinear}. Second, a fully-serial implementation with $P=1$ results in a latency for SSC decoding that scales as $(2+o(1)) N\log_2\log_2 N$. Third, it is shown that $P=N^{1/\mu}$ in a semi-parallel implementation is the smallest $P$ that results in the same latency as that of the fully-parallel implementation of SSC decoding. %Future work includes finding the exact bounds for the latency of SSC decoding of polar codes in the low-rate regime.

%----------------------------------------------------
\section*{Acknowledgments}
%----------------------------------------------------

S.~A.~Hashemi is supported by a Postdoctoral Fellowship from the Natural Sciences and Engineering Research Council of Canada (NSERC) and by Huawei. M.~Mondelli is partially supported by the 2019 Lopez-Loreta Prize. A.~Fazeli and A.~Vardy were supported in part by the National Science Foundation under Grant CCF-1764104.

% Generated by IEEEtran.bst, version: 1.14 (2015/08/26)
\newcommand{\SortNoop}[1]{}

\appendix 

\section{Proofs}\label{app:prf}

\begin{proof}[Proof of Lemma \ref{lemma:unpolarized}]
By applying Lemma~1 in \cite{mondelli2020sublinear}, for $n_0\ge 1$,
\begin{equation}\label{eq:nummid1}
    \mathbb P(Z_{n_0}\in [2^{- 2n_0}, 1-2^{- 2n_0}]) \le c_1\,2^{-n_0/\mu},
\end{equation}
where $c_1$ is a universal constant which does not depend on $n_0$, $W$.
Let $\{B_n\}_{n\ge 1}$ be a sequence of i.i.d. random variables with distribution Bernoulli$\left(1/2\right)$. Then, by using \eqref{eq:eqBMSC}, it is clear that, for $n\ge 1$,
\begin{equation*}
Z_{n_0+n} \le  \left\{\begin{array}{ll} Z_{n_0+n-1}^2, & \mbox{ if } B_n=1, \\ 2 Z_{n_0+n-1}, & \mbox{ if } B_n=0. \end{array}\right. 
\end{equation*}
Therefore, by applying Lemma 22 of \cite{HAU14}, we obtain that, for $n_1\ge 1$,
\begin{equation}\label{eq:lemma22}
{\mathbb P}\left(Z_{n_0+n_1} \le 2^{\scriptstyle-2^{\scriptstyle\sum_{i=1}^{n_1} B_i}}\mid Z_{n_0} = x\right)\ge 1- c_2 \hspace{0.1em} x(1-\log_2 x),
\end{equation}
with $c_2 = 2/(\sqrt{2}-1)^2$. Thus,
\begin{equation}\label{eq:Zn0n11}
\begin{split}
{\mathbb P}\left(Z_{n_0+n_1} \le 2^{\scriptstyle-2^{\scriptstyle\sum_{i=1}^{n_1} B_i}}\mid Z_{n_0} \le 2^{-n_0}\right)  & \ge  1- c_2 \hspace{0.1em} 2^{-n_0}(1+n_0) \\
& \ge  1- c_2\hspace{0.1em} \frac{\sqrt{2}}{\ln 2} \hspace{0.1em} 2^{-n_0/\mu}, 
\end{split}
\end{equation}
where the first inequality uses the fact that $1- c_2 \hspace{0.1em} x(1-\log_2 x)$ is decreasing in $x$ for any $x\le 2^{-n_0}\le 1/2$, and the second inequality uses that $1- c_2 \hspace{0.1em} 2^{-n_0}(1+n_0)\ge 1- c_2 \hspace{0.1em} \sqrt{2}\cdot 2^{-n_0/2} / \ln 2 $ for any $n_0\in \mathbb N$ and that $\mu >2$. Furthermore, by using the same passages of (54) in \cite{MHU15unif-ieeeit}, we obtain that, for any $\epsilon \in (0, 1/2)$,
\begin{equation}\label{eq:Zn0n12}
{\mathbb P}\left(2^{\scriptstyle-2^{\scriptstyle\sum_{i=1}^{n_1} B_i}} > 2^{\scriptstyle-2^{\scriptstyle n_1 \epsilon}}\right) \le 2^{-n_1(1-h_2(\epsilon))},
\end{equation}
where $h_2(x)=-x\log_2 x-(1-x)\log_2 (1-x)$ denotes the binary entropy function. By combining \eqref{eq:Zn0n11} and \eqref{eq:Zn0n12},
\begin{equation}\label{eq:Zn0n1fin}
    {\mathbb P}\left(Z_{n_0+n_1} \le 2^{\scriptstyle-2^{\scriptstyle n_1 \epsilon}}\mid Z_{n_0} \le 2^{-n_0}\right)   \ge 1- c_2\hspace{0.1em} \frac{\sqrt{2}}{\ln 2} \hspace{0.1em} 2^{-n_0/\mu}-2^{-n_1(1-h_2(\epsilon))}.
\end{equation}

Define $Y_n=1-Z_n$. Note that, if $Z_{n+1}=Z_n^2$, then 
\begin{equation}
    Y_{n+1}=1-(1-Y_n)^2=2Y_n-Y_n^2\le 2Y_n.
\end{equation}
Furthermore, if $Z_{n+1}\ge Z_n\sqrt{2-Z_n^2}$, then
\begin{equation}
    Y_{n+1}\le 1-(1-Y_n)\sqrt{2-(1-Y_n)^2}\le 2 Y_n^2,
\end{equation}
where in the last inequality the fact that $1-t\sqrt{2-t^2}\le 2(1-t)^2$ for any $t\in [0, 1]$ is used. Thus, by using \eqref{eq:eqBMSC}, for $n\ge 1$,
\begin{equation*}
Y_{n_0+n} \le  \left\{\begin{array}{ll} 2Y_{n_0+n-1}^2, & \mbox{ if } B_n=1, \\ 2 Y_{n_0+n-1}, & \mbox{ if } B_n=0. \end{array}\right. 
\end{equation*}
Define $\tilde{Y}_{n_0}=2Y_{n_0}$ and
\begin{equation*}
\tilde{Y}_{n_0+n} =  \left\{\begin{array}{ll} \tilde{Y}_{n_0+n-1}^2, & \mbox{ if } B_n=1, \\ 2 \tilde{Y}_{n_0+n-1}, & \mbox{ if } B_n=0. \end{array}\right. 
\end{equation*}
Then for any $n\ge 0$,
\begin{equation}\label{eq:tildeYY}
    Y_{n_0+n}\le \frac{1}{2}\tilde{Y}_{n_0+n}\le \tilde{Y}_{n_0+n}.
\end{equation}
By applying again Lemma~22 of \cite{HAU14} to the process $\tilde{Y}_n$, for $n_1\ge 1$,
\begin{equation}\label{eq:Zn0n13}
\begin{split}
{\mathbb P}\left(\tilde{Y}_{n_0+n_1} \le 2^{\scriptstyle-2^{\scriptstyle\sum_{i=1}^{n_1} B_i}}\mid \tilde{Y}_{n_0} \le 2^{-n_0}\right)  & \ge 1- c_2\hspace{0.1em} \frac{\sqrt{2}}{\ln 2} \hspace{0.1em} 2^{-n_0/\mu}, 
\end{split}
\end{equation}
which, combined with \eqref{eq:Zn0n12}, gives that, for any $\epsilon \in (0, 1/2)$,
\begin{equation}\label{eq:tildeYbd}
    {\mathbb P}\left(\tilde{Y}_{n_0+n_1} \le 2^{\scriptstyle-2^{\scriptstyle n_1 \epsilon}}\mid \tilde{Y}_{n_0} \le 2^{-n_0}\right)   \ge 1- c_2\hspace{0.1em} \frac{\sqrt{2}}{\ln 2} \hspace{0.1em} 2^{-n_0/\mu}-2^{-n_1(1-h_2(\epsilon))}.
\end{equation}
By using \eqref{eq:tildeYY} and the fact that $\tilde{Y}_{n_0}=2Y_{n_0}$, \eqref{eq:tildeYbd} implies that 
\begin{equation}\label{eq:Yn0n1fin}
        {\mathbb P}\left(Y_{n_0+n_1} \le 2^{\scriptstyle-2^{\scriptstyle n_1 \epsilon}}\mid Y_{n_0} \le 2^{-n_0-1}\right)   \ge 1- c_2\hspace{0.1em} \frac{\sqrt{2}}{\ln 2} \hspace{0.1em} 2^{-n_0/\mu}-2^{-n_1(1-h_2(\epsilon))}.
\end{equation}

Let $n\ge 1$. Set $n_1 = \lceil\gamma n\rceil$, $n_0 = n - \lceil\gamma n\rceil$, and $\epsilon = h_2^{(-1)}\left((\gamma(\mu+1)-1)/(\gamma\mu)\right)$, where $h_2^{(-1)}(\cdot)$ is the inverse of $h_2(x)$ for any $x\in [0, 1/2]$. Note that if $\gamma \in \left(1/(1+\mu), 1\right)$, then $\epsilon\in (0, 1/2)$. Consequently, \eqref{eq:Zn0n1fin} implies that
\begin{equation}\label{eq:Zn0n1fin2}
{\mathbb P}\left(Z_{n} \le 2^{- \scriptstyle 2^{\scriptstyle n \hspace{0.1em}\gamma\hspace{0.1em} h_2^{(-1)}\left(\frac{\gamma(\mu+1)-1}{\gamma\mu}\right)}}\mid Z_{n_0} \le 2^{-n_0}\right) 
\ge 1- c_3 \hspace{0.1em} 2^{-n\scriptstyle\frac{1-\gamma}{\mu}},
\end{equation}
where $c_3$ is a numerical constant. Similarly, by using that $Z_n=1-Y_n$, from \eqref{eq:Yn0n1fin},
\begin{equation}\label{eq:Yn0n1fin2}
{\mathbb P}\left(Z_{n} \ge 1- 2^{- \scriptstyle 2^{\scriptstyle n \hspace{0.1em}\gamma\hspace{0.1em} h_2^{(-1)}\left(\frac{\gamma(\mu+1)-1}{\gamma\mu}\right)}}\mid Z_{n_0} \ge 1- 2^{-n_0-1}\right) 
\ge 1- c_3 \hspace{0.1em} 2^{-n\scriptstyle\frac{1-\gamma}{\mu}}.
\end{equation}
The proof is concluded by the following chain of inequalities:
{\allowdisplaybreaks
\begin{align*}
    %\begin{split}
        \mathbb P&\left(Z_n\in \left[2^{-2^{n\gamma h_2^{(-1)}\left(\frac{\gamma(\mu+1)-1}{\gamma\mu}\right)}}, 1-2^{-2^{n\gamma h_2^{(-1)}\left(\frac{\gamma(\mu+1)-1}{\gamma\mu}\right)}}\right]\right)\\
        &= 1- \mathbb P\left(Z_n\le 2^{-2^{n\gamma h_2^{(-1)}\left(\frac{\gamma(\mu+1)-1}{\gamma\mu}\right)}}\right)- \mathbb P\left(Z_n\ge 1-2^{-2^{n\gamma h_2^{(-1)}\left(\frac{\gamma(\mu+1)-1}{\gamma\mu}\right)}}\right) \\
        &\le 1- \mathbb P\left(Z_n\le 2^{-2^{n\gamma h_2^{(-1)}\left(\frac{\gamma(\mu+1)-1}{\gamma\mu}\right)}}, Z_{n_0} \le 2^{-n_0}\right)\\
        &\hspace{2em}- \mathbb P\left(Z_n\ge 1-2^{-2^{n\gamma h_2^{(-1)}\left(\frac{\gamma(\mu+1)-1}{\gamma\mu}\right)}}, Z_{n_0} \ge 1- 2^{-n_0-1}\right)\\
        &= 1- \mathbb P\left(Z_n\le 2^{-2^{n\gamma h_2^{(-1)}\left(\frac{\gamma(\mu+1)-1}{\gamma\mu}\right)}}\mid Z_{n_0} \le 2^{-n_0}\right)\mathbb P\left(Z_{n_0} \le 2^{-n_0}\right)\\
        &\hspace{2em}- \mathbb P\left(Z_n\ge 1-2^{-2^{n\gamma h_2^{(-1)}\left(\frac{\gamma(\mu+1)-1}{\gamma\mu}\right)}}\mid Z_{n_0} \ge 1- 2^{-n_0-1}\right)\mathbb P\left(Z_{n_0} \ge 1- 2^{-n_0-1}\right)\\
        &\stackrel{\mathclap{\mbox{\footnotesize(a)}}}{\le} 1- \left(1- c_3 \hspace{0.1em} 2^{-n\scriptstyle\frac{1-\gamma}{\mu}}\right)\left(1-\mathbb P(Z_{n_0}\in [2^{- n_0}, 1-2^{- n_0-1}])\right)\\
        &\stackrel{\mathclap{\mbox{\footnotesize(b)}}}{\le} 1- \left(1- c_3 \hspace{0.1em} 2^{-n\scriptstyle\frac{1-\gamma}{\mu}}\right)\left(1-c_1\,2^{-n\scriptstyle\frac{1-\gamma}{\mu}}\right)\\
        &\le (c_3+c_1) \hspace{0.1em} 2^{-n\scriptstyle\frac{1-\gamma}{\mu}},
    %\end{split}
\end{align*}
}
where \eqref{eq:Zn0n1fin2} and \eqref{eq:Yn0n1fin2} are used in (a), and \eqref{eq:nummid1} is used in (b).
\end{proof}

\end{document}